\documentclass[a4paper,fleqn]{cas-sc}

\usepackage[sort&compress,square,numbers]{natbib}
\usepackage[capitalize]{cleveref}
\usepackage{aliascnt}
\usepackage{braket}
\usepackage{amsthm}
\usepackage{algorithm, algorithmic}
\def\tsc#1{\csdef{#1}{\textsc{\lowercase{#1}}\xspace}}
\tsc{WGM}
\tsc{QE}
\tsc{EP}
\tsc{PMS}
\tsc{BEC}
\tsc{DE}
\newcommand{\Or}{\mathcal{O}}

\DeclareMathOperator{\polylog}{polylog}

\newtheorem{theorem}{Theorem}
\crefname{theorem}{Theorem}{Theorems}
\Crefname{theorem}{Theorem}{Theorems}

\newaliascnt{lemma}{theorem}
\newtheorem{lemma}[lemma]{Lemma}
\aliascntresetthe{lemma}
\crefname{lemma}{Lemma}{Lemmas}
\Crefname{lemma}{Lemma}{Lemmas}

\newaliascnt{proposition}{theorem}
\newtheorem{proposition}[proposition]{Proposition}
\aliascntresetthe{proposition}
\crefname{proposition}{Proposition}{Propositions}
\Crefname{proposition}{Proposition}{Propositions}

\newaliascnt{corollary}{theorem}
\newtheorem{corollary}[corollary]{Corollary}
\aliascntresetthe{corollary}
\crefname{corollary}{Corollary}{Corollaries}
\Crefname{corollary}{Corollary}{Corollaries}

\newdefinition{rmk}{Remark}

\ExplSyntaxOn
\cs_set:Npn \__first_footerline:
{
  \group_begin:
  \small
  \sffamily
  \ifnum\theblind>0\relax
  \else
    \__short_authors: :~
  \fi
  { \rmfamily \itshape Preprint }
  \group_end:
}
\ExplSyntaxOff

\begin{document}
\let\WriteBookmarks\relax
\def\floatpagepagefraction{1}
\def\textpagefraction{.001}
\let\printorcid\relax

\shorttitle{Circuit Depth Reduction of One-Ancilla QDE Solver}
\shortauthors{D. Fang and J. Park}

\title[mode = title]{Circuit Depth Reduction of One-Ancilla Quantum Differential Equation Solver via Extrapolation}

\author[1]{Di Fang}

\affiliation[1]{organization={Department of Mathematics and Duke Quantum Center, Duke University},
                city={Durham},
                state={NC},
                country={USA}}

\author[2]{Justin Park}

\affiliation[2]{organization={Department of Mathematics and Department of Computer Science, Duke University},
                city={Durham},
                state={NC},
                country={USA}}

\begin{abstract}
Solving linear differential equations is a fundamental task in scientific computing and an important primitive for quantum computing.
A recent one-ancilla quantum differential equation solver provides a hardware-friendly and locality-preserving approach with provable performance guarantees, making it highly suitable for the early fault-tolerant and near-term regimes. Its simple circuit structure comes with a natural trade-off: the maximum single-run circuit depth scales as $\Or
(1/\epsilon)$ in the target accuracy $\epsilon$. In this work, we reduce this depth by combining the solver with classical 
step-size postprocessing.
By running the one-ancilla solver at a logarithmic number of finite time step sizes and using classical post-processing to cancel leading discretization errors, we reduce the maximum single-run circuit depth to $\Or(\polylog(1/\epsilon))$ without adding quantum ancillae or sacrificing locality. Technically, extending extrapolation ideas beyond Hamiltonian and Lindbladian dynamics requires regularity estimates for observable maps under nonunitary evolution, which we obtain through a holomorphic extension of the adjoint evolution. 
Numerical experiments on the Hatano-Nelson model (ODE) and the convection-diffusion equation (PDE) demonstrate the effectiveness of the approach.
\end{abstract}


\maketitle

\section{Introduction}

Differential equations are among the most important mathematical models in science and engineering, with wide-ranging applications across many disciplines. Designing quantum algorithms for differential equations has been an important direction in quantum computing. In particular, time-dependent linear differential equations have received significant attention, especially in dissipative and non-unitary regimes where the dynamics cannot be directly treated as standard Hamiltonian simulation. Significant developments in quantum algorithms have been made for this task. While we do not attempt to give an exhaustive survey here, representative state-of-the-art approaches include methods based on quantum linear systems algorithms with history-state encodings~\cite{Berry2014,BerryChildsOstranderWang2017,ChildsLiu2020,Krovi2022,BerryCosta2024,AnOnwuntaYang2024}, time-marching methods~\cite{FangLinTong2023,YangOnwuntaAn2025}, linear combinations of Hamiltonian simulation~\cite{AnLiuLin2023linear,AnChildsLinYing2024laplace,AnChildsLin2023quantum,NovikauJoseph2025,PocrnicJohnsonKatabarwaWiebe2025,KharaziAlkadriMandadapuWhaley2026}, unitary-dilation~\cite{Li2025,WuLi2026,JinLiuYu2024,JinLiuMaPengYu2025}, Lindbladian embeddings~\cite{ShangGuoAnZhao2024}, quantum eigenvalue processing~\cite{LowSu2024}, and other fault-tolerant primitives. We remark that this dissipative regime is a standard setting shared across state-of-the-art quantum algorithmic frameworks for general linear differential equations, rather than an additional assumption specific to the present work or to a particular line of work.

As in many areas of quantum algorithms, existing quantum differential equation solvers can be broadly viewed from two complementary perspectives. On the one hand, near-term algorithms are designed to be more compatible with current or early fault-tolerant hardware, often avoiding expensive coherent control, large ancilla registers, or deep circuit constructions. Variational and other NISQ approaches fall into this category, but rigorous performance guarantees are often difficult to establish in general. On the other hand, fully fault-tolerant quantum algorithms are developed with provable guarantees and can achieve excellent, sometimes optimal and near-optimal asymptotic scaling. However, these algorithms typically rely on advanced fully fault-tolerant primitives such as block-encodings, quantum linear systems algorithms (QLSA), quantum singular value transformation (QSVT), linear-combination-of-unitaries (LCU) techniques, compression gadget, or coherent control over many time steps. These ingredients can require many ancilla qubits and complicated circuit structures, making them less friendly to realize on near-term or early fault-tolerant devices.

A recent algorithm in~\cite{FangGeorgeTong2025} provides a different point in this landscape. It gives a one-ancilla quantum algorithm for linear differential equations of the dissipative form considered in this work. The algorithm only uses local Hamiltonian simulation and mid-circuit measurements, preserves locality, and avoids advanced fault-tolerant subroutines such as LCU, QSVT, QLSA, or compression gadget. In this sense, it is both near-term friendly and mathematically rigorous: it is simple enough to be implemented with an optimal one-ancilla qubit, yet still comes with a provable error analysis. This makes it an attractive candidate for early demonstrations of quantum algorithms for non-unitary dynamics and differential equations.
This favorable hardware structure, however, comes with a tradeoff. The algorithm does not aim to achieve the optimal asymptotic scaling of fully fault-tolerant quantum differential equation solvers. In particular, its circuit depth has an $\Or(1/\epsilon)$ dependence on the target accuracy $\epsilon$, reflecting the first-order nature of the underlying quantum algorithm.

Quantum resources remain expensive, even in the early fault-tolerant regime. It is therefore important to use classical computation and processing whenever possible to reduce the amount of coherent quantum evolution required. A particularly successful strategy in this direction is to use classical post-processing, and in particular extrapolation-based algorithmic error mitigation, to reduce quantum resources without changing the underlying quantum simulation primitive. Pioneering works in this direction include~\cite{RendonWatkinsWiebe2024,WatsonWatkins2024,Watson2024}, which explored this idea for Hamiltonian simulation, including circuit-depth reduction for Trotter product formulas and qDRIFT, as well as recent works on open quantum systems governed by Lindbladian dynamics~\cite{MohammadipourLi2025,MohammadipourLi2026}. The step-size extrapolation viewpoint is also closely related to zero-noise extrapolation in quantum error mitigation, where data obtained at several effective noise strengths are classically extrapolated to the zero-noise limit~\cite{TemmeBravyiGambetta2017,LiBenjamin2017,EndoBenjaminLi2018,CaiEtAl2023}.

In this work, we bring this perspective to the one-ancilla quantum algorithm for linear differential equations. The dynamics considered here are not necessarily unitary or completely positive and trace-preserving, as in Hamiltonian or Lindbladian cases. We ask whether extrapolation-based classical post-processing can provably reduce the circuit depth of this near-term-friendly differential equation solver, without changing the underlying quantum circuit or introducing additional ancilla qubits. We answer this question affirmatively for unnormalized observable estimation. For each time step size, we run the one-ancilla circuit to the final time and estimate an unnormalized observable from raw measurement data. We then combine the estimates from multiple step sizes using either Richardson extrapolation or Chebyshev interpolation, thereby extrapolating the observable to the zero-step-size limit. In this way, tools from classical numerical analysis and approximation theory are used to extract a more accurate observable estimate from a family of coarse, hardware-friendly quantum circuits.

Technically, there are two important differences between the present setting and the previously studied Hamiltonian and Lindblad settings. The first difference concerns the regularity estimates for the step-size-dependent observable map, the crucial step for this kind of postprocessing analysis. In Hamiltonian simulation, observable dynamics can be described naturally in the Heisenberg picture by unitary conjugation. In a trace-preserving Lindblad simulation, there is likewise a natural dual semigroup acting on observables. These structures allow one to analyze the step-size dependence of observables directly. In the present postselected ODE setting, however, the successful evolution is described by a non-unitary map acting on state vectors, and the target observable has the form of a quadratic expression involving both the forward map and its Hilbert-space adjoint. When the one-step map is analytically continued to a complex step size, taking the Hilbert-space adjoint becomes anti-holomorphic rather than holomorphic. To obtain derivative bounds for extrapolation, we therefore introduce a holomorphic adjoint lift, which agrees with the usual adjoint on real step sizes but remains holomorphic in the complexified step-size variable. This is the key analytic ingredient that allows us to prove the regularity estimates needed for Richardson extrapolation and Chebyshev interpolation.

The second difference concerns the sampling and observable-estimation model. In Hamiltonian and trace-preserving Lindblad simulation, the simulated object is itself a normalized quantum state or density matrix, and each circuit execution contributes a bounded observable sample from that object. For general non-unitary differential equations, by contrast, a successful quantum output is a normalized state that encodes the solution rather than the unnormalized solution itself, while the squared norm of the unnormalized solution is encoded in the postselection success probability. In the present postselected ODE algorithm, each raw circuit execution includes mid-circuit postselection measurements as part of its built-in design; if one of these measurements returns an undesired outcome, the trajectory is rejected by the postselection procedure. In the original state-preparation view, such rejected trajectories are discarded by design. For unnormalized observable estimation, however, we retain the statistical contribution of these rejection events by recording them as zero-valued samples, while postselection-accepted trajectories contribute the measured observable outcome. This zero-on-rejection estimator is unbiased for the unnormalized observable associated with the time-discretized solution. Consequently, the success probability is automatically incorporated into the raw estimator, and no separate estimation or extrapolation of the success probability is required.

Our main contribution is a complete extrapolation framework for this setting, which reduces the maximum quantum circuit depth from an $\Or(1/\epsilon)$ dependence on the target accuracy $\epsilon$ to an $\Or(\polylog(1/\epsilon))$ dependence (\cref{thm:richardsonResources}, \cref{thm:chebyshevResources}). We prove derivative estimates for the unnormalized observable as a function of the time step size (\cref{thm:derivBounds}), establish deterministic bias bounds for general linear post-processing (\cref{thm:biasBound}), and then specialize the framework to Richardson extrapolation and Chebyshev interpolation, deriving for each the sample complexity, circuit depth, and gate count (\cref{thm:richardsonResources}, \cref{thm:chebyshevResources}). 
We also give an upper bound on the sample complexity of the zero-on-failure estimator via concentration inequalities (\cref{thm:samplingCost}), together with a matching lower bound for the independent-shot sampling model (\cref{thm:samplingLowerBound}), showing the resource count to be optimal. Under a geometrically local assumption on the generator, these bounds specialize to an explicit maximum circuit depth in terms of the system size and evolution time (\cref{cor:maxDepth}), and to fully explicit gate count scalings (\cref{cor:richFinal}, \cref{cor:chebFinal}).
In particular, for Chebyshev interpolation, the maximum circuit depth has only logarithmic-squared dependence on the target accuracy for observable estimation, replacing the inverse-accuracy dependence of the direct first-order method; Richardson extrapolation achieves an even smaller logarithmic dependence (\cref{cor:maxDepth}). However, Richardson extrapolation's weights grow polynomially in $1/\epsilon$, rather than Chebyshev's polylogarithmic growth (\cref{cor:weightBounds}). In geometrically local settings, this gives a depth reduction while preserving the one-ancilla and locality-preserving structure of the original algorithm.
We also discuss numerical tests that illustrate the theory. One class of examples comes from convection-diffusion equations, where the differential operator naturally decomposes into anti-self-adjoint and dissipative parts and where the jump operators are simple in either the grid basis or the Fourier basis. Another class comes from the interacting Hatano-Nelson model, a non-Hermitian many-body system exhibiting the non-Hermitian skin effect. These examples demonstrate how the extrapolation framework can be applied both to PDE-inspired problems and to many-body non-Hermitian dynamics.

The rest of the paper is organized as follows. In~\cref{sec:algorithm}, we review the one-ancilla quantum differential equation solver and present the hybrid quantum-classical extrapolation algorithm. \cref{sec:analysis} proves the local error bounds, the backward-error derivative estimates, and the deterministic extrapolation error bounds. In~\cref{sec:resources}, we analyze the sampling concentration and quantum resource requirements. \cref{sec:numerics} presents numerical experiments for convection-diffusion equations and the interacting Hatano-Nelson model. We conclude in~\cref{sec:conclusion}.

\section{Algorithm Overview and Classical Postprocessing}\label{sec:algorithm}
This section sets up the algorithmic framework used throughout the paper. We first recall the one-ancilla quantum differential equation solver, describe how to classically post-process the measurement outcomes, and then provide a pseudocode for the algorithm.

\subsection{Algorithm}\label{subsec:algorithm}
Throughout, we work in a finite dimension and consider the autonomous dissipative linear ODE 
\begin{equation}
    \frac{d}{dt}\ket{\psi(t)}=A\ket{\psi(t)},\quad A=-iH-\sum_{j=1}^{J}D_j, \quad D_j=L_j^\dagger L_j\succeq 0,
\end{equation}
with Hermitian $H=H^\dagger$ and $J$ dissipative terms. The exact solution is $\ket{\psi(t)}=e^{tA}\ket{\psi(0)}$. We assume $\ket{\psi_0}=\ket{\psi(0)}$ is a normalized initial state, $\|\ket{\psi_0}\|=1$.

We recall the one-ancilla algorithm construction of ~\cite{FangGeorgeTong2025}, establishing notation for the rest of the paper. For each dissipative term $D_j$, define the Hermitian dilation
\begin{equation}
    G_j=\begin{pmatrix}
        0 & L_j^\dagger \\
        L_j & 0
    \end{pmatrix}.
\end{equation}
Post-selecting the ancilla on outcome $\ket{0}$ after evolving under $G_j$ for time $\sqrt{2s}$ gives the system-only operator
\begin{equation}
    C_j(s) := (\bra{0} \otimes I) e^{i\sqrt{2s}G_j}(\ket{0}\otimes I)=I-sD_j+O(s^2),
\end{equation}
and the leading order term reproduces the dissipative contraction $e^{-sD_j}$. The closed form of this postselected block is $C_j(s)=\cos(\sqrt{2s}D_j^{1/2})$, which is proved in \cref{lem:cosineBlock}. Composing the $J$ postselected blocks with one Hamiltonian simulation step approximates the evolution $e^{sA}$ over a single time step of size $s$:
\begin{equation}
    K_s:=e^{-isH}C_J(s)\cdots C_1(s).
\end{equation}
For a realizable step size $s=T/R$, $R\in\mathbb{N}$, the successful state after $R$ steps is 
\begin{equation}
    \ket{u_s(T)}:=K_s^R\ket{\psi_0},
\end{equation}
with a success probability $p(s):=\|\ket{u_s(T)}\|^2$. As $s\rightarrow 0$, $\ket{u_s(T)}$ approximates $\ket{\psi(T)}$.

For an observable $O=O^\dagger$, we aim to estimate the observable signal 
\begin{equation}
    g_O(0):= \bra{\psi(T)}O\ket{\psi(T)}=\bra{\psi_0}(e^{TA})^{\dagger}O(e^{TA})\ket{\psi_0}.
\end{equation}
The data produced by a single circuit run with a step size $s$ is 
\begin{equation}
    g_O(s):= \bra{u_s(T)}O\ket{u_s(T)},
\end{equation}
and the extrapolation algorithm uses $g_O(s_\ell)$ for several distinct step sizes $s_\ell$ to estimate $g_O(0)$. If the desired target is instead the observable's value conditional on successful postselection, one could either extrapolate $f_O(s)=\frac{g_O(s)}{p(s)}$ directly from successful postselections, or extrapolate $g_O(s)$ and $p(s)$ separately and then take a ratio.

\subsection{The Extrapolation Method}\label{subsec:extrapolation}
\cref{subsec:localEB} shows that achieving accuracy $\epsilon$ directly, by shrinking the step size $s$ until the algorithm's own discretization error is small enough, forces a circuit depth scaling as $\mathcal{O}(1/\epsilon)$. The rest of this paper shows that this cost can be reduced to $\mathcal{O}(\mathrm{polylog}(1/\epsilon))$ by running the circuit at several step sizes that remain coarse, and combining the results through classical post-processing.

For a fixed step size $s$, the circuit is run repeatedly. Each run returns the random variable 
\begin{equation}
    Y_s:=\begin{cases}
        \text{the measured value of $O$ on the successful branch}, &\text{if every postselection succeeds}\\
        0, &\text{otherwise}
    \end{cases},
\end{equation}
recording zero rather than discarding the run when a postselection fails, so that $\mathbb{E}[Y_s]=g_O(s)$ (\cref{lem:unbiasedness}). Averaging $Y_s$ over many repetitions gives an estimate of $g_O(s)$. \cref{subsec:samplingModel} bounds the number of repetitions needed for a target accuracy.

Running this procedure at several distinct step sizes $s_1,\ldots, s_m$, each one realizable and each coarser than direct estimation at the target accuracy would require, gives a sequence of estimates of $g_O(s_1),\ldots, g_O(s_m)$. These are combined into a single estimate of $g_O(0)$ using a weighted average $\sum_\ell a_\ell g_O(s_\ell)$, with weights chosen so that the combination agrees with $g_O(0)$ up to an error controlled by how smoothly $g_O$ depends on $s$.

Richardson extrapolation and Chebyshev interpolation, treated in \cref{subsec:richardson} and \cref{subsec:chebyshev}, are two choices of step sizes and weights for this same procedure. Richardson extrapolation places step sizes at harmonically shrinking values, while Chebyshev interpolation clusters them according to the roots of Chebyshev polynomials. The two trade off differently: Richardson achieves a smaller maximum circuit depth for a given accuracy target, while Chebyshev's weights are better conditioned, which affects the number of circuit repetitions needed rather than the depth itself.

\subsection{Pseudocode}\label{subsec:pseudocode}
The quantum sampling routine is the same for Richardson and Chebyshev post-processing. The only input that changes is the list of nodes and weights $\{(s_\ell,a_\ell)\}_{\ell=1}^{m}$. The pseudocode for the algorithm is shown in \cref{alg:observableEstimation}. 

\begin{algorithm}[h]
    \caption{Extrapolation-Based Observable Estimation}\label{alg:observableEstimation}
    \begin{algorithmic}[1]
        \STATE Input: final time $T$, observable $O$, target accuracy $\epsilon$, failure probability $\delta$, realizable step sizes $s_\ell=T/R_\ell$, and extrapolation weights $a_\ell$, $\ell=1,\ldots, m$.
        \STATE Compute $\Lambda_m=\sum_{\ell=1}^{m}|a_\ell|$ and set $\eta=\epsilon/(2\Lambda_m)$.
        \FOR{$\ell=1,\ldots, m$}
            \FOR{$w=1,\ldots, N$}
                \STATE Run the circuit with step size $s_\ell$ and $R_\ell=T/s_\ell$ steps to final time $T$.
                \IF{All postselections succeed}
                    \STATE Measure $O$ and record the measurement outcome.
                \ELSE
                    \STATE Record the value zero.
                \ENDIF
            \ENDFOR
            \STATE Average recorded $N$ values to obtain $\hat{g}_O(s_\ell)$, which estimates $g_O(s_\ell)$ to accuracy $\eta$ with failure probability at most $\delta/m$. 
        \ENDFOR
        \STATE Output $\hat{g}_O(0)=\sum_{\ell=1}^{m}a_\ell\hat{g}_O(s_\ell)$.
    \end{algorithmic}
\end{algorithm}

\section{Numerical Analysis}\label{sec:analysis}
This section proves the analytic estimates underlying the postprocessing procedure. We begin with local error bounds for the one-ancilla solver, and then use holomorphic adjoint arguments to obtain derivative estimates for the step-size-dependent observable map. These derivative estimates are the key ingredient in the extrapolation analysis and lead to the deterministic extrapolation bias bounds.

\subsection{Local Error Bounds}\label{subsec:localEB}
We bound the one-step defect $K_s-e^{sA}$ and its accumulation over $R$ steps to the final time $T$. Write 
\begin{equation}
    B_0:=-iH,\quad B_j=-D_j\;\text{for}\; 1\le j\le J, \quad A=\sum_{\nu=0}^{J}B_{\nu},
\end{equation}
and define
\begin{equation}
    S:=\sum_{\nu=0}^{J}\|B_\nu\|=\|H\|+\sum_{j=1}^{J}\|D_j\|,
\end{equation}
\begin{equation}
    \Gamma_{\mathrm{PF}}:=\sum_{0\le \mu<\nu\le J}\|[B_\nu,B_\mu]\|, \quad \Gamma_{\mathrm{cos}}:=\frac{2}{3}\sum_{j=1}^{J}\|D_j^2\|, \quad \Gamma_{\mathrm{loc}}:=\frac{1}{2}\Gamma_{\mathrm{PF}}+\Gamma_{\mathrm{cos}}.
\end{equation}
$S$ measures the overall strength of the generator. The two error scales have different origins: $\Gamma_{\mathrm{PF}}$ comes from the noncommutativity of $J+1$ pieces $B_\nu$, while $\Gamma_{\mathrm{cos}}$ is the truncation error of each postselected cosine block. $\Gamma_{\mathrm{loc}}$ combines them into the single scale used in the bounds below.

Before using these constants, we record a closed-form expression for $C_j(s)$ that the remainder of this subsection relies on.
\begin{lemma}[Cosine Block]\label{lem:cosineBlock}
    For every real $s\ge 0$ and each $j=1,\ldots, J$,
    \begin{equation}
        C_j(s)=\cos(\sqrt{2s}D_j^{1/2}).
    \end{equation}
\end{lemma}
\begin{proof}
    Since $G_j$ is purely off-diagonal in the ancilla block, $G_j^2$ is a block-diagonal matrix:
    \begin{equation}
        G_j^2
        =\begin{pmatrix}
            L_j^\dagger L_j & 0\\
            0 & L_jL_j^\dagger
        \end{pmatrix}
        =\begin{pmatrix}
            D_j & 0\\
            0 & L_jL_j^\dagger
        \end{pmatrix}.
    \end{equation}
    So, $(\bra{0}\otimes I)G_j^{2k}(\ket{0}\otimes I)=D_j^k$ while $(\bra{0}\otimes I)G_j^{2k+1}(\ket{0}\otimes I)=0$. 
    Thus projecting the power series for $e^{i\sqrt{2s}G_j}$ onto the ancilla state $\ket{0}$ kills every odd term and leaves
    \begin{equation}
        C_j(s)=(\bra{0}\otimes I)e^{i\sqrt{2s}G_j}(\ket{0}\otimes I)=\sum_{k=0}^{\infty}\frac{(-1)^k(2s)^k}{(2k)!}D_j^k=\cos(\sqrt{2s}D_j^{1/2}).
    \end{equation}
\end{proof}

This closed form allows us to read off the operator norm of $C_j(s)$ directly from the spectrum of $D_j$. 
\begin{lemma}[Contractivity]\label{lem:contractivity}
    For every real $s\ge0$,
    \begin{equation}
        \|e^{sA}\|\le 1,\quad \|C_j(s)\|\le 1, \quad \|K_s\|\le 1.
    \end{equation}
\end{lemma}
\begin{proof}
    For any vector $\ket{v}$, the function $h(s)=\|e^{sA}\ket{v}\|^2$ satisfies 
    \begin{equation}
        h'(s)=(e^{sA}\ket{v})^\dagger(A+A^\dagger)(e^{sA}\ket{v}) \le 0,
    \end{equation}
    because $A+A^\dagger=-2\sum_{j=1}^{J}D_j\preceq0$ is negative semidefinite. This means that $h(s)\le h(0)$, i.e. $\|e^{sA}\ket{v}\|\le \|\ket{v}\|$. Therefore, $\|e^{sA}\|\le 1$ for all real $s\ge 0$.

    Moreover, by the normality of a positive semidefinite matrix $D_j$, the spectral theorem implies that $D_j=U\mathrm{diag}(\lambda_1,\ldots, \lambda_n)U^\dagger$ for some unitary \(U\) and real nonnegative eigenvalues \(\lambda_i\ge 0\). So,
    \begin{equation}
        C_j(s)=U\mathrm{diag}(\cos(\sqrt{2s\lambda_1}),\ldots,\cos(\sqrt{2s\lambda_n}))U^\dagger,
    \end{equation}
    and thus 
    \begin{equation}
        \|C_j(s)\|\le \max_i|\cos(\sqrt{2s\lambda_i})|\le 1.
    \end{equation}
    
    Lastly, since $e^{-isH}$ is unitary, $\|K_s\|\le \|e^{-isH}\|\cdot \prod_{j=1}^{J}\|C_j(s)\|\le 1$.
\end{proof}

Since $C_j(s)$ and $e^{-sD_j}$ are both contractions, we next bound the distance between them directly.
\begin{lemma}[Cosine Block Defect]\label{lem:cosineDefect}
    For every real $s\ge 0$, 
    \begin{equation}
        \|C_j(s)-e^{-sD_j}\|\le \frac{2}{3}s^2\|D_j^2\|.
    \end{equation}
\end{lemma}
\begin{proof}
    From Taylor expansion of $C_j(s)$ with integral remainder, 
    \begin{equation}\label{eq:realCosTaylor}
        C_j(s)=I-sD_j+R_C(s),\quad R_C(s)=\frac{2}{3}s^2D_j^2\int_{0}^{1}(1-t)^3\cos(t\sqrt{2s}D_j^{1/2})\,dt.
    \end{equation}
    Similar to the proof of \cref{lem:contractivity}, $\|\cos(t\sqrt{2s}D_j^{1/2})\|\le 1$ for $t\in[0,1], s\ge 0$, so 
    \begin{equation}
        \|R_C(s)\|\le \frac{2}{3}s^2\|D_j^2\|\int_{0}^{1}(1-t)^3\,dt=\frac{s^2}{6}\|D_j^2\|.
    \end{equation}

    Similarly, for $e^{-sD_j}$, the Taylor expansion is 
    \begin{equation}\label{eq:realExpTaylor}
        e^{-sD_j}=I-sD_j+R_E(s), \quad R_E(s)=s^2D_j^2\int_0^1(1-t)e^{-tsD_j}\,dt.
    \end{equation}
    Since $D_j\succeq0$, $\|e^{-tsD_j}\|\le 1$, and hence
    \begin{equation}
        \|R_E(s)\|\le s^2\|D_j^2\|\int_{0}^{1}(1-t)\,dt=\frac{s^2}{2}\|D_j^2\|.
    \end{equation}

    Then, by the triangle inequality,
    \begin{equation}
        \|C_j(s)-e^{-sD_j}\|\le \|R_C(s)\|+\|R_E(s)\|\le \frac{2}{3}s^2\|D_j^2\|.
    \end{equation}
\end{proof}

Combined with the contraction bounds above, this per-block defect yields a bound on the error of the full one-step map $K_s$ against $e^{sA}$.
\begin{proposition}[One-Step Defect]\label{prop:stepDefect}
    For every real $s\ge 0$, 
    \begin{equation}
        \|K_s-e^{sA}\|\le s^2e^{sS}\Gamma_{\mathrm{loc}}.
    \end{equation}
    In particular, if $sS\le \log 2$, then $\|K_s-e^{sA}\|\le 2s^2\Gamma_{\mathrm{loc}}$
\end{proposition}
\begin{proof}
    Let
    \begin{equation}
        Q_s:= e^{sB_0}e^{sB_J}\cdots e^{sB_1}.
    \end{equation}
    The error bound for the first-order Lie-Trotter product-formula~\cite{ChildsSuTranEtAl2020,AnFangLin2021} gives
    \begin{equation}\label{eq:realPF}
        \|Q_s-e^{sA}\| \le \frac{s^2}{2}e^{s\sum_{\nu=0}^{J}\|B_\nu\|}\sum_{0\le \mu<\nu\le J}\|[B_\nu,B_\mu]\|= \frac{s^2}{2}e^{sS}\Gamma_{\mathrm{PF}}.
    \end{equation}

    Next, using the telescoping identity,
    \begin{equation}
        K_s-Q_s=e^{sB_0}\sum_{j=1}^{J}C_J(s)\cdots C_{j+1}(s)(C_j(s)-e^{-sD_j})e^{-sD_{j-1}}\cdots e^{-sD_1}.
    \end{equation}
    Then, by \cref{lem:contractivity} and \cref{lem:cosineDefect},
    \begin{equation}
        \|K_s-Q_s\|\le \sum_{j=1}^{J}\|C_j(s)-e^{-sD_j}\|\le\frac{2}{3}s^2\sum_{j=1}^{J}\|D_j^2\|=s^2\Gamma_{\mathrm{cos}}.
    \end{equation}

    By the triangle inequality, 
    \begin{equation}
        \|K_s-e^{sA}\| \le s^2\Gamma_{\mathrm{cos}}+\frac{s^2}{2}e^{sS}\Gamma_{\mathrm{PF}}\le s^2e^{sS}\Gamma_{\mathrm{loc}},
    \end{equation}
    since $e^{sS}\ge 1$. In particular, if $sS\le \log 2$, then $e^{sS}\le 2$, thus $\|K_s-e^{sA}\|\le 2s^2\Gamma_{\mathrm{loc}}$.
\end{proof}

Having controlled the error of a single step, we accumulate it over $R$ repetitions to reach the final time $T$.
\begin{proposition}[Global First-Order Defect]\label{prop:globalDefect}
    Let $s=T/R$, $R\in\mathbb{N}$. Then,
    \begin{equation}
        \|K_s^R-e^{TA}\|\le Rs^2e^{sS}\Gamma_{\mathrm{loc}},
    \end{equation}
    and in particular, if $sS\le \log 2$, $\|K_s^R-e^{TA}\|\le 2Ts\Gamma_{\mathrm{loc}}$.
\end{proposition}
\begin{proof}
    The telescoping identity gives
    \begin{equation}
        K_s^R-e^{TA}=\sum_{r=0}^{R-1}K_s^r(K_s-e^{sA})e^{sA(R-1-r)}.
    \end{equation}
    By \cref{lem:contractivity} and \cref{prop:stepDefect},
    \begin{equation}
        \|K_s^R-e^{TA}\|\le \sum_{r=0}^{R-1}\|K_s-e^{sA}\|\le Rs^2e^{sS}\Gamma_{\mathrm{loc}}.
    \end{equation}
    In particular, if $sS\le \log 2$, then 
    \begin{equation}
        \|K_s^R-e^{TA}\|\le 2Rs^2\Gamma_{\mathrm{loc}}=2Ts\Gamma_{\mathrm{loc}},
    \end{equation}
    using the fact that $s=T/R$. 
\end{proof}

Achieving accuracy $\epsilon$ therefore requires step size $s=\mathcal{O}(\epsilon/(T\Gamma_{\mathrm{loc}}))$, or equivalently $R=\mathcal{O}(T^2\Gamma_{\mathrm{loc}}/\epsilon)$ time steps, giving the $\mathcal{O}(1/\epsilon)$ circuit depth that motivates the extrapolation framework developed in the remainder of the paper.

\subsection{Complex Step Error Bounds}\label{subsec:complexEB}
To bound how $g_O(s)$ varies with step size, we extend $C_j(s)$ and $K_s$ to complex step size $z$ and study the resulting holomorphic objects. This extension is required for the extrapolation bias estimates established later in this section. The constants $S$ and $\Gamma_{\mathrm{loc}}$ from \cref{subsec:localEB} reappear unchanged. However, contractivity no longer holds off the real axis, so the bounds below grow exponentially in $|z|$ rather than remaining bounded by one.

For complex $z$, define
\begin{equation}
    C_j(z):=\sum_{k=0}^{\infty}\frac{(-1)^k(2z)^k}{(2k)!}D_j^k, \quad K_z:=e^{zB_0}C_J(z)\cdots C_1(z).
\end{equation}
Then $K_z$ is entire in $z$, with $K_z=I+zA+O(z^2)$. 

We first bound the defect of each postselected block against the exponential it approximates on the complex domain.
\begin{lemma}[Complex Cosine Block Defect]\label{lem:complexCosDefect}
    There exist universal constants $c_0,C_0>0$ such that, for every complex $z$ with $|z|S\le c_0$ and each $j=1,\ldots, J$,
    \begin{equation}
        \|C_j(z)-e^{-zD_j}\|\le C_0|z|^2\|D_j^2\|.
    \end{equation}
\end{lemma}
\begin{proof}
    From Taylor expansion with integral remainder, valid for any complex $z$,
    \begin{equation}
        C_j(z)=I-zD_j+R_C(z),\quad R_C(z)=\frac{2}{3}z^2D_j^2\int_{0}^{1}(1-t)^3\cos(t\sqrt{2z}D_j^{1/2})\,dt,
    \end{equation}
    \begin{equation}
        e^{-zD_j}=I-zD_j+R_E(z), \quad R_E(z)=z^2D_j^2\int_0^1(1-t)e^{-tzD_j}\,dt.
    \end{equation}

    Writing $\cos(W)=\frac{1}{2}(e^{iW}+e^{-iW})$ for $W=t\sqrt{2z}D_{j}^{1/2}$ and bounding each exponential via $\|e^X\|\le e^{\|X\|}$ gives $\|\cos(t\sqrt{2z}D_j^{1/2})\|\le e^{\sqrt{2|z|S}}$ for $t\in[0,1]$, since $\|D_j\|\le S$. Similarly $\|e^{-tzD_j}\|\le e^{|z|S}$. Hence,
    \begin{equation}
        \|R_C(z)\|\le \frac{2}{3}|z|^2\|D_j^2\|e^{\sqrt{2|z|S}}\int_{0}^{1}(1-t)^3\,dt\le  \frac{1}{6}|z|^2\|D_j^2\|e^{\sqrt{2|z|S}},
    \end{equation}
    \begin{equation} 
        \|R_E(z)\| \le |z|^2\|D_j\|^2e^{|z|S}\int_0^1(1-t)\,dt \le \frac{1}{2}|z|^2\|D_j^2\|e^{|z|S}.
    \end{equation}
    Since the leading order terms cancel, $C_j(z)-e^{-zD_j}=R_C(z)-R_E(z)$, so by triangle inequality and $|z|S\le c_0$,
    \begin{equation}
        \|C_j(z)-e^{-zD_j}\| \le \|R_C(z)\|+\|R_E(z)\|\le (\frac{1}{6}e^{\sqrt{2c_0}}+\frac{1}{2}e^{c_0})|z|^2\|D_j^2\|=: C_0|z|^2\|D_j^2\|.
    \end{equation}
\end{proof}

Then, we collect the exponential bounds that replace contractivity for complex step size.
\begin{lemma}[Bounded Norms for Complex Step Size]\label{lem:boundedNorms}
    There exist universal constants $c_1,C_1>0$ with $c_1\le c_0$ such that, for every complex $z$ with $|z|S\le c_1$,
    \begin{equation}
        \|e^{zB_\nu}\|\le e^{|z|\|B_\nu\|}\;(0\le \nu\le J), \quad \|C_j(z)\|\le e^{C_1|z|\|D_j\|}\;(1\le j\le J).
    \end{equation}
\end{lemma}
\begin{proof}
    For any operator $X$ and complex constant $z$, 
    \begin{equation}
        \|e^{zX}\|=\left\|\sum_{k=0}^{\infty}\frac{(zX)^k}{k!}\right\|\le \sum_{k=0}^{\infty}\frac{|z|^k\|X\|^k}{k!}=e^{|z|\|X\|}.
    \end{equation}
    So, $\|e^{zB_\nu}\|\le e^{|z|\|B_\nu\|}$ holds for all $\nu$. 

    For $C_j(z)$, write 
    \begin{equation}
        \|C_j(z)\|\le \|e^{-zD_j}\|+\|C_j(z)-e^{-zD_j}\|\le e^{|z|\|D_j\|}+C_0|z|^2\|D_j\|^2,
    \end{equation}
    using \cref{lem:complexCosDefect} and $\|D_j^2\|\le \|D_j\|^2$. Since $|z|\|D_j\|\le |z|S|\le c_1$, we have $|z|^2\|D_j\|^2\le c_1\cdot |z|\|D_j\|$, so 
    \begin{equation}
        \|C_j(z)\|\le e^{|z|\|D_j\|}(1+C_0c_1|z|\|D_j\|)\le e^{|z|\|D_j\|}e^{C_0c_1|z|\|D_j\|}=e^{(1+C_0c_1)|z|\|D_j\|},
    \end{equation}
    using $1+x\le e^x$. Take $C_1:=1+C_0c_1$ so that $\|C_j(z)\|\le e^{C_1|z|\|D_j\|}$. 
\end{proof}

The defect of $K_z$ against $e^{zA}$ now follows the same argument as \cref{prop:stepDefect}, with each contractivity bound replaced by its complex counterpart from the previous lemma.
\begin{lemma}[Complex One-Step Defect]\label{lem:complexStepDefect}
    There exist universal constants $c_2,C_2>0$ with $c_2\le c_1$ such that, whenever $|z|S\le c_2$, 
    \begin{equation}
        \|K_z-e^{zA}\|\le C_2|z|^2\Gamma_{\mathrm{loc}}.
    \end{equation}
\end{lemma}
\begin{proof}
    Let $Q_z=e^{zB_0}e^{zB_J}\cdots e^{zB_1}$. The complex first-order Lie-Trotter product formula bound gives
    \begin{equation}
        \|Q_z-e^{zA}\|\le C|z|^2e^{C|z|S}\Gamma_{\mathrm{PF}}\le C|z|^2\Gamma_{\mathrm{PF}}
    \end{equation}
    for $|z|S\le c_2$, absorbing the bounded exponential factor into the constant. By the telescoping identity,
    \begin{equation}
        K_z-Q_z = e^{zB_0}\sum_{j=1}^{J}C_J(z)\cdots C_{j+1}(z)(C_j(z)-e^{-zD_j})e^{-zD_{j-1}}\cdots e^{-zD_1}.
    \end{equation}
    By \cref{lem:boundedNorms}, the surrounding factors satisfy
    \begin{equation}
        \|e^{zB_0}\|\cdot \|C_J(z)\cdots C_{j+1}(z)\|\cdot \|e^{-zD_{j-1}}\cdots e^{-zD_1}\|\le e^{|z|\|H\|}\cdot e^{C_1|z|\sum_{k=j+1}^{J}\|D_k\|}\cdot e^{|z|\sum_{k=1}^{j-1}\|D_k\|}\le e^{(1+C_1)|z|S}.
    \end{equation}
    since each partial sum of $\|D_k\|$ is bounded by $S$ directly, with no dependence on $J$ beyond what $S$ already carries. Combining this with \cref{lem:complexCosDefect},
    \begin{equation}
        \|K_z-Q_z\|\le e^{(1+C_1)|z|S}\sum_{j=1}^{J}C_0|z|^2\|D_j^2\|\le e^{(1+C_1)c_2}\cdot \frac{3}{2}C_0|z|^2\Gamma_{\mathrm{cos}},
    \end{equation}
    using $|z|S\le c_2$ to bound the exponential factor and $\Gamma_{\mathrm{cos}}=\frac{2}{3}\sum_j\|D_j^2\|$. Adding the two estimates and absorbing constants,
    \begin{equation}
        \|K_z-e^{zA}\|\le \|K_z-Q_z\|+\|Q_z-e^{zA}\| \le C|z|^2\Gamma_{\mathrm{cos}}+ C|z|^2\Gamma_{\mathrm{PF}}=:C_2|z|^2\Gamma_{\mathrm{loc}}.
    \end{equation}
\end{proof}

This defect bound controls $K_z$ itself. However, extrapolation requires control of the generator. We therefore show that $K_z$ admits a well-defined logarithm near $z=0$ and bound the distance of this logarithm from $A$.
\begin{lemma}[Modified Generator Bound]\label{lem:generatorBound}
    There exist universal constants $c_3,C_3>0$ with $c_3\le c_2$ such that, whenever $|z|S\le c_3$, the principal logarithm of $K_z$ is well-defined, and 
    \begin{equation}
        A_{\mathrm{eff}}(z):=z^{-1}\log K_z
    \end{equation}
    extends holomorphically to $z=0$ by setting $A_{\mathrm{eff}}(0):=A$. 
    Moreover,
    \begin{equation}\label{eq:generatorBound}
        \|A_{\mathrm{eff}}(z)-A\|\le C_3|z|\Gamma_{\mathrm{loc}}.
    \end{equation}
\end{lemma}
\begin{proof}
    For $|z|S\le c_3$ with $c_3\le c_2$ sufficiently small, both $e^{zA}$ and $K_z$ lie in the ball $\|X-I\|\le \frac{1}{2}$. For $e^{zA}$, it satisfies
    \begin{equation}\label{eq:smallExp}
        \|e^{zA}-I\|\le e^{|z|\|A\|}-1\le e^{|z|S}-1\le e^{c_3}-1,
    \end{equation}
    which is small. For $K_z$, we can bound using \cref{lem:complexStepDefect} and \cref{eq:smallExp},
    \begin{equation}
        \|K_z-I\|\le \|K_z-e^{zA}\|+\|e^{zA}-I\|\le C|z|^2\Gamma_{\mathrm{loc}}+(e^{c_3}-1),
    \end{equation}
    and both terms vanish as $c_3\rightarrow 0^+$ because the first term vanishes by $|z|\rightarrow 0$ as $|z|\le c_3/S$ and the second term is proven small.
    
    Since $K_z$ and $e^{zA}$ both lie in this ball, the principal logarithm is well-defined and Lipschitz there, so the bound on $\|K_z-e^{zA}\|$ converts directly into a bound on $\|\log K_z-zA\|$. For $z\ne 0$, dividing it by $|z|$, we get
    \begin{equation}
        \|A_{\mathrm{eff}}(z)-A\|\le C|z|\Gamma_{\mathrm{loc}}.
    \end{equation}
    It remains to show that $A_{\mathrm{eff}}(z)=z^{-1}\log K_z$ extends holomorphically across $z=0$. Since $\|K_z-I\|\le 1/2$ on this disk, the Mercator series gives
    \begin{equation}
        \log K_z=(K_z-I)-\frac{1}{2}(K_z-I)^2+\frac{1}{3}(K_z-I)^3+\cdots,
    \end{equation}
    valid and absolutely convergent there. Using $K_z-I=zA+\mathcal{O}(z^2)$, and noting that $(K_z-I)^2$ is already $\mathcal{O}(z^2)$, every term beyond the first contributes only at order $z^2$ or higher, so
    \begin{equation}
        \log K_z=zA+\mathcal{O}(z^2).
    \end{equation}
    In particular, $\log K_z$ has no constant term, so $A_{\mathrm{eff}}=z^{-1}\log K_z$ extends holomorphically to $z=0$ by setting $A_{\mathrm{eff}}(0):=A$, matching the coefficient of $z$ above. \cref{eq:generatorBound} holds at $z=0$ as well, since both sides vanish there.
\end{proof}

\subsection{Holomorphic Adjoint Lift}\label{subsec:adjointLift}
The eventual bound on the observable signal requires both $A_{\mathrm{eff}}(z)$ and an adjoint object to be holomorphic in $z$ simultaneously. Writing $A_{\mathrm{eff}}(z)=A+\sum_{k=1}^{\infty}E_kz^k$ for its Taylor expansion about $z=0$, guaranteed by \cref{lem:generatorBound}, the operator adjoint $A^\dagger _{\mathrm{eff}}(z)$ conjugates both the coefficients $E_k$ and the scalar $z^k$, giving
\begin{equation}
    A^\dagger _{\mathrm{eff}}(z) = A^\dagger +\sum_{k\ge 1}E_k^\dagger \bar{z}^k.
\end{equation}
This expression depends on $\bar{z}$ rather than $z$, so it is anti-holomorphic, not holomorphic, off the real axis. 

To retain holomorphy, we instead conjugate only the operator coefficients, leaving the scalar $z^k$ untouched, and define the holomorphic adjoint lift
\begin{equation}
    A_{\mathrm{eff}}^\sharp (z):=A^\dagger +\sum_{k\ge 1}E_k^\dagger z^k.
\end{equation}
For real $z=s$, $\bar{s}=s$, so $A^\sharp_{\mathrm{eff}}(s)=A_{\mathrm{eff}}(s)^\dagger$ exactly. For complex $z$, $A^\sharp_{\mathrm{eff}}(z)$ remains holomorphic.

\begin{lemma}[Adjoint-Lift Bound]\label{lem:adjointBound}
    Let $\rho>0$ satisfy $\rho S\le c_3$, so that the bound in \cref{eq:generatorBound} holds for $|z|\le \rho$. Then, there exists a universal constant $C_4>0$ such that, for every $z$ with $|z|\le \rho/2$, 
    \begin{equation}
        \|A^\sharp_{\mathrm{eff}}(z)-A^\dagger\|\le C_4|z|\Gamma_{\mathrm{loc}}.
    \end{equation}
\end{lemma}
\begin{proof}
    Define
    \begin{equation}
        F(z):= \begin{cases}
            \frac{A_{\mathrm{eff}}(z)-A}{z} &z\ne 0\\E_1& z=0
        \end{cases}.
    \end{equation}
    Because of the removable singularity at $z=0$, $F(z)=\sum_{k\ge 1}E_kz^{k-1}$ and \cref{eq:generatorBound} yields
    \begin{equation}
        \frac{\|A_{\mathrm{eff}}(z)-A\|}{|z|}=\|F(z)\| \le C_3\Gamma_{\mathrm{loc}}.
    \end{equation}
    Applying Cauchy's estimate to the Taylor coefficients of $F$, $\|E_{k+1}\|\le C_3\Gamma_{\mathrm{loc}}\cdot \rho^{-k}$. Therefore,
    \begin{equation}
        \|A^\sharp_{\mathrm{eff}}(z)-A^\dagger\|\le \sum_{k\ge 1}\|E_k\||z|^k\le C_3\Gamma_{\mathrm{loc}}\cdot |z|\sum_{k\ge 1}\left(\frac{|z|}{\rho}\right)^{k-1},
    \end{equation}
    and at $|z|\le \rho/2$, the geometric series converges to at most $2$, which yields $\|A^\sharp_{\mathrm{eff}}(z)-A^\dagger\|\le C_4|z|\Gamma_{\mathrm{loc}}$.
\end{proof}

\subsection{Sharp Derivative Bounds for the Observable Signal}
The complex step constructions of the previous two subsections yield derivative bounds on $g_O(s)$ as a function of the real step size $s$, required by the extrapolation bias bound later. Cauchy's integral formula converts a bound on a holomorphic function over a disk into a bound on its derivatives at the center of a smaller disk. It remains to show that the relevant holomorphic lift of $g_O$ is bounded on such a disk. 

Define $Q:=\max\{S,T\Gamma_{\mathrm{loc}}\}$. This constant combines the strength of the generator with the accumulated effect of dissipation over the full evolution time.

A preliminary identity motivates the construction below. By definition, $K_s=e^{sA_{\mathrm{eff}}(s)}$, so for the realizable step size $s=T/R$,
\begin{equation}
    K_s^R=e^{RsA_{\mathrm{eff}}(s)}=e^{TA_{\mathrm{eff}}(s)}.
\end{equation}
Since $\ket{u_s(T)}=K_s^R\ket{\psi_0}$, this gives $\ket{u_s(T)}=e^{TA_{\mathrm{eff}}(s)}\ket{\psi_0}$, which means the holomorphic lift constructed below reproduces $g_O(s)$ precisely at real, realizable step sizes.

\begin{theorem}[Sharp Derivative Bounds for the Observable Signal]\label{thm:derivBounds}
    There exist universal constants $c,C,C'>0$ such that the following holds. Let $0\le s\le s_Q:=c/Q$. Then, for every integer $m\ge 0$,
    \begin{equation}
        \left|\frac{d^m}{ds^m}g_O(s)\right|\le C\|O\|(C'Q)^m m!, \quad \left|\frac{d^m}{ds^m}p(s)\right|\le C(C'Q)^m m!.
    \end{equation}
    Here $g_O(s)$ and $p(s)$ denote the real restrictions of the holomorphic lifts constructed in the proof, which agree with the circuit quantities at realizable step sizes.
\end{theorem}
\begin{proof}
    \textbf{Step 1 (logarithmic norm bound).} Let $X$ be a bounded operator, and define the logarithmic norm $\mu(X):=\lambda_{\mathrm{max}}(\frac{X+X^\dagger}{2})$. Then, for any vector $\ket{v}$ and $t\ge 0$, 
    \begin{equation}
        \frac{d}{dt}\|e^{tX}\ket{v}\|^2 = \bra{v}(e^{tX})^\dagger(X+X^\dagger)(e^{tX})\ket{v}\le 2\mu(X)\|e^{tX}\ket{v}\|^2,
    \end{equation}
    using $\frac{X+X^\dagger}{2}\preceq \mu(X)I$. By Gronwall's inequality, $\|e^{tX}\ket{v}\|^2\le e^{2t\mu(X)}\|\ket{v}\|^2$, so the logarithmic norm controls the growth of its exponential by $\|e^{tX}\|\le e^{t\mu(X)}$ for $t\ge 0$.

    Write $A_{\mathrm{eff}}(z)=A+E(z)$, where the bound for $E(z):=A_{\mathrm{eff}}(z)-A$ is shown in \cref{lem:generatorBound}. The logarithmic norm of $A_{\mathrm{eff}}(z)$ is bounded as follows
    \begin{equation}
        \mu(A_{\mathrm{eff}}(z)) = \lambda_{\mathrm{max}}\left(\frac{A+A^\dagger}{2}+\frac{E(z)+E^\dagger(z)}{2}\right)\le \mu(A)+\frac{1}{2}\sup_{\|\ket{v}\|=1}\bra{v}(E(z)+E^\dagger(z))\ket{v}\le \mu(A)+\|E(z)\|.
    \end{equation}
    Then, by the dissipativity of $A$ where $A+A^\dagger\preceq 0$, $\mu(A)\le 0$ gives $\mu(A_{\mathrm{eff}}(z))\le \|E(z)\|$ and hence
    \begin{equation}
        \|e^{TA_{\mathrm{eff}}(z)}\|\le e^{T\|E(z)\|}.
    \end{equation}

    Bound $\|E(z)\|\le C_3|z|\Gamma_{\mathrm{loc}}$ via \cref{lem:generatorBound} for $|z|S\le c_3$. Invoking \cref{lem:adjointBound} requires a radius $\rho$ with $\rho S\le c_3$, so set $\rho:=4c/Q$ so that this holds once $c\le c_3/4$, using $S\le Q$. On the disk $|z|\le \rho/2=2c/Q$, 
    \begin{equation}
        T\|E(z)\|\le T \cdot C_3\Gamma_{\mathrm{loc}}\cdot \frac{2c}{Q}=2C_3c\cdot \frac{T\Gamma_{\mathrm{loc}}}{Q}\le 2C_3c,
    \end{equation}
    using $T\Gamma_{\mathrm{loc}}\le Q$. Choosing $c\le \min\{\frac{c_3}{4}, \frac{1}{2C_3}\}$ gives $T\|E(z)\|\le 1$, and hence $\|e^{TA_{\mathrm{eff}}(z)}\|\le e$. The identical argument, with $A^\dagger$ in place of $A$ since $\mu(A^\dagger)=\mu(A)$, bounds $\|e^{TA^\sharp_{\mathrm{eff}}(z)}\|\le e$ using \cref{lem:adjointBound}.

    \textbf{Step 2 (bounded holomorphic lifts)}. Define 
    \begin{equation}
        G_O(z):= \bra{\psi_0}e^{TA^\sharp_{\mathrm{eff}}(z)}Oe^{TA_{\mathrm{eff}}(z)}\ket{\psi_0}, \quad P(z):= \bra{\psi_0}e^{TA^\sharp_{\mathrm{eff}}(z)}e^{TA_{\mathrm{eff}}(z)}\ket{\psi_0}.
    \end{equation}
    By Cauchy-Schwarz and Step 1,
    \begin{equation}
        |G_O(z)|\le \|e^{TA^{\sharp}_{\mathrm{eff}}(z)}\|\cdot \|O\|\cdot \|e^{TA_{\mathrm{eff}}(z)}\|\le C\|O\|,
    \end{equation}
    and similarly $|P(z)|\le C$, on $|z|\le \rho/2$.

    \textbf{Step 3 (Cauchy's formula).} For real $s\in[0,s_Q]$ with $s_Q=c/Q=\rho/4$, the disk $|z-s|\le \rho/4$ is contained in $|z|\le \rho/2$, since $|z|\le |z-s|+s\le \rho/4+\rho/4=\rho/2$. Cauchy's formula for the $m$-th derivative of a function bounded by $C\|O\|$ on a disk of radius $\rho/4$ gives
    \begin{equation}
        |G_O^{(m)}(s)|\le C\|O\|\frac{m!}{(\rho/4)^m}=C\|O\|m!\left(\frac{Q}{c}\right)^m,
    \end{equation}
    using $\rho/4=c/Q$. Setting $C':=1/c$ gives the stated form, and the bound on $P^{(m)}(s)$ follows identically:
    \begin{equation}
        \left|\frac{d^m}{ds^m}G_O(s)\right|\le C\|O\|\left(C'Q\right)^mm!, \quad  \left|\frac{d^m}{ds^m}P(s)\right|\le C\left(C'Q\right)^mm!.
    \end{equation}

    \textbf{Step 4 (identification).} By the preliminary identity above, $\ket{u_s(T)}=e^{TA_{\mathrm{eff}}(s)}\ket{\psi_0}$ at real, realizable $s$, and $A^\sharp_{\mathrm{eff}}(s)=A_{\mathrm{eff}}(s)^\dagger$ by construction of the adjoint lift. Substituting,
    \begin{equation}
        G_O(s)=\bra{\psi_0}e^{TA_{\mathrm{eff}}(s)^\dagger}Oe^{TA_{\mathrm{eff}}(s)}\ket{\psi_0}=\bra{u_s(T)}O\ket{u_s(T)}=g_O(s),
    \end{equation}
    and likewise
    \begin{equation}
        P(s)=\bra{\psi_0}e^{TA_{\mathrm{eff}}(s)^\dagger}e^{TA_{\mathrm{eff}}(s)}\ket{\psi_0}=\braket{u_s(T)|u_s(T)}=\|\ket{u_s(T)}\|^2=p(s),
    \end{equation}
    so the bounds from Step 3 transfer directly to the quantities stated.
\end{proof}

\subsection{Deterministic Extrapolation Biases}
The derivative bound established in \cref{thm:derivBounds} converts directly into a bias bound for any linear combination of $g_O(s_\ell)$ values, through the classical Lagrange interpolation remainder formula.

Let $s_1,\ldots, s_m$ be distinct realizable step sizes in $(0,s_{\mathrm{max}}]$. Define the Lagrange basis polynomials and the resulting extrapolation weights to zero
\begin{equation} \label{eq:def_a_l_Lambda_m}
    \phi_\ell(x):=\prod_{k\ne \ell}\frac{x-s_k}{s_\ell-s_k}, \quad a_\ell:=\phi_\ell(0), \quad \Lambda_m:=\sum_{\ell=1}^{m}|a_\ell|.
\end{equation}
The extrapolated estimator is 
\begin{equation}
    \hat{g}_O(0):=\sum_{\ell=1}^{m}a_\ell\hat{g}_O(s_\ell),
\end{equation}
where $\hat{g}_O(s_\ell)$ denotes a raw estimator of $g_O(s_\ell)$ obtained from circuit measurements.

\begin{theorem}[Deterministic Extrapolation Bias]\label{thm:biasBound}
    Suppose $s_\ell\in[0,s_Q]$ for every $\ell$, where $s_Q$ is as in \cref{thm:derivBounds}. Then
    \begin{equation}
        \left|\sum_{\ell=1}^{m}a_\ell g_O(s_\ell)-g_O(0)\right|\le C\|O\|(C'Q)^m\prod_{\ell=1}^{m}s_\ell.
    \end{equation}
\end{theorem}
\begin{proof}
    Since $\sum_{\ell=1}^{m}a_\ell g_O(s_\ell)$ is, by construction, the value at $x=0$ of the unique degree at most $(m-1)$ polynomial interpolating $g_O$ at $s_1,\ldots, s_m$, the classical Lagrange interpolation remainder formula gives
    \begin{equation}
        g_O(0)-\sum_{\ell=1}^{m}a_\ell g_O(s_\ell)=\frac{g_O^{(m)}(\xi)}{m!}\prod_{\ell=1}^{m}(0-s_\ell),
    \end{equation}
    for some $\xi$ in the interval spanned by $0,s_1,\ldots,s_m$, hence $\xi\in[0,s_Q]$. Applying \cref{thm:derivBounds} at $\xi$,
    \begin{equation}
        \left|g_O^{(m)}(\xi)\right|\le C\|O\|(C'Q)^m m!,
    \end{equation}
    so
    \begin{equation}
        \left|g_O(0)-\sum_{\ell=1}^{m}a_\ell g_O(s_\ell)\right|\le \frac{C\|O\|(C'Q)^m m!}{m!}\prod_{\ell=1}^{m} |s_\ell| = C\|O\|(C'Q)^m\prod_{\ell=1}^{m}s_\ell,
    \end{equation}
    using $s_\ell\ge 0$.
\end{proof}

\section{Concentration and Quantum Resource Estimate}\label{sec:resources}
This section establishes the quantum resource estimates. We first use concentration inequalities to derive sampling-cost bounds for the zero-on-failure estimator, and then obtain maximum single-run circuit-depth and total gate-count estimates for both Richardson extrapolation and Chebyshev interpolation.

\subsection{Sampling Model and Concentration}\label{subsec:samplingModel}
\cref{thm:biasBound} bounds the deterministic error of extrapolating from exact values $g_O(s_\ell)$. In practice each $g_O(s_\ell)$ is itself only available as an estimate from a finite number of circuit executions, using the random variable $Y_s$ defined in \cref{subsec:extrapolation}. We bound the resulting statistical error and the number of executions required to control it.

\begin{lemma}[Unbiasedness]\label{lem:unbiasedness}
    $\mathbb{E}[Y_s]=g_O(s)$, and $|Y_s|\le \|O\|$.
\end{lemma}
\begin{proof}
    Conditional on success, the postselected state is $\frac{\ket{u_s(T)}}{\|\ket{u_s(T)}\|}$, so the conditional expectation of the measured value is 
    \begin{equation}
        \frac{\bra{u_s(T)}O\ket{u_s(T)}}{\|\ket{u_s(T)}\|^2}=\frac{g_O(s)}{p(s)}.
    \end{equation}
    Since success occurs with probability $p(s)$,
    \begin{equation}
        \mathbb{E}[Y_s]=p(s)\cdot \frac{g_O(s)}{p(s)}+(1-p(s))\cdot 0 =g_O(s).
    \end{equation}

    The bound $|Y_s|\le \|O\|$ holds because $Y_s$ is either $0$ or an eigenvalue of $O$ realized by the measurement, both bounded by $\|O\|$.
\end{proof}

Boundedness and unbiasedness are the hypotheses required to apply Hoeffding's inequality, used below to bound the sampling error.
\begin{theorem}[Sampling Cost]\label{thm:samplingCost}
    Let $s_1,\ldots, s_m$ and $a_1,\ldots, a_m$ be as in \cref{thm:biasBound}, with $\Lambda_m:=\sum_\ell|a_\ell|$. Fix a target accuracy $\epsilon>0$ and failure probability $\delta\in(0,1)$. For each node $\ell$, draw
    \begin{equation}
        N:=\left\lceil \frac{8\|O\|^2\Lambda_m^2}{\epsilon^2}\log\frac{2m}{\delta} \right\rceil
    \end{equation}
    independent samples of $Y_{s_\ell}$, written $Y^{(1)}_{s_\ell},\ldots, Y^{(N)}_{s_\ell}$, and let 
    \begin{equation}
        \hat{g}_O(s_\ell):=\frac{1}{N}\sum_{w=1}^{N}Y^{(w)}_{s_\ell}
    \end{equation}
    be their average. Then,
    \begin{equation}
        \mathbb{P}\left(\left|\sum_{\ell=1}^{m}a_\ell \hat{g}_O(s_\ell)-\sum_{\ell=1}^{m}a_\ell g_O(s_\ell)\right|>\frac{\epsilon}{2}\right)\le \delta
    \end{equation}
\end{theorem}
\begin{proof}
    Fix a node $\ell$. By \cref{lem:unbiasedness}, $\mathbb{E}[Y^{(w)}_{s_\ell}]=g_O(s_\ell)$ for each $w$, and $Y^{(w)}_{s_\ell}\in[-\|O\|,\|O\|]$ almost surely. Since $Y^{(1)}_{s_\ell},\cdots, Y^{(N)}_{s_\ell}$ are independent and identically distributed, Hoeffding's inequality applied to their average $\hat{g}_O(s_\ell)$ gives, for any $\varphi>0$,
    \begin{equation}
        \mathbb{P}(|\hat{g}_O(s_\ell)-g_O(s_\ell)|\ge \varphi)\le 2\exp\left(-\frac{N\varphi^2}{2\|O\|^2}\right).
    \end{equation}
    Set $\varphi:=\epsilon/(2\Lambda_m)$. Substituting the stated value of $N$,
    \begin{equation}
        \frac{N\varphi^2}{2\|O\|^2}\ge \frac{1}{2\|O\|^2}\cdot\frac{8\|O\|^2\Lambda_m^2}{\epsilon^2}\log\frac{2m}{\delta}\cdot \frac{\epsilon^2}{4\Lambda_m^2}=\log\frac{2m}{\delta},
    \end{equation}
    so
    \begin{equation}
        2\exp\left(-\frac{N\varphi^2}{2\|O\|^2}\right)\le 2\exp\left(-\log \frac{2m}{\delta}\right)=2\cdot \frac{\delta}{2m}=\frac{\delta}{m}.
    \end{equation}

    Hence $\mathbb{P}(|\hat{g}_O(s_\ell)-g_O(s_\ell)|\ge \varphi)\le \delta/m$ for each $\ell$. By the union bound over the $m$ nodes, with probability at least $1-\delta$, $|\hat{g}_O(s_\ell)-g_O(s_\ell)|< \varphi$ simultaneously for every $\ell=1,\ldots ,m$. On this event, the triangle inequality gives
    \begin{equation}
        \left|\sum_{\ell=1}^{m}a_\ell\hat{g}_O(s_\ell)-\sum_{\ell=1}^{m}a_\ell g_O(s_\ell)\right|\le \sum_{\ell=1}^{m}|a_\ell|\cdot |\hat{g}_O(s_\ell)-g_O(s_\ell)|<\Lambda_m\varphi=\frac{\epsilon}{2}.
    \end{equation}
\end{proof}

\cref{thm:biasBound} and \cref{thm:samplingCost} combine into the total error guarantee for the extrapolated estimator.
\begin{corollary}[Total Error Bound]\label{cor:totalError}
    Suppose $s_1,\ldots, s_m\in[0,s_Q]$ are chosen so that the deterministic bias from \cref{thm:biasBound} is at most $\epsilon/2$, and let $N$ be as in \cref{thm:samplingCost}. Then
    \begin{equation}
        \mathbb{P}\left(\left|\sum_{\ell=1}^{m}a_\ell \hat{g}_O(s_\ell)-g_O(0)\right|>\epsilon\right)\le \delta
    \end{equation}
\end{corollary}
\begin{proof}
    By the triangle inequality,
    \begin{equation}
        \left|\sum_{\ell=1}^{m}a_\ell\hat{g}_O(s_\ell)-g_O(0)\right| \le \left|\sum_{\ell=1}^{m}a_\ell\hat{g}_O(s_\ell)-\sum_{\ell=1}^{m}a_\ell g_O(s_\ell)\right| + \left|\sum_{\ell=1}^{m}a_\ell{g}_O(s_\ell)-g_O(0)\right|.
    \end{equation}
    The second term is the deterministic bias, at most $\epsilon/2$ by hypothesis. By \cref{thm:samplingCost}, the first term exceeds $\epsilon/2$ with probability at most $\delta$. Hence, the sum exceeds $\epsilon$ with probability at most $\delta$.
\end{proof}

\subsection{Resource Estimate}\label{subsec:resource}
\cref{thm:samplingCost} gives the number of circuit executions $N$ required for each step size. We convert this into circuit depth and gate count. For a step size $s_\ell$, one raw circuit execution to the final time $T$ has depth
\begin{equation}
    D_{\mathrm{run}}(s_\ell):=\frac{T}{s_\ell}d_{\mathrm{step}},
\end{equation}
where $d_{\mathrm{step}}$ is the depth of a single time step. If $g_\mathrm{step}$ is the gate count of one time step and $g_{\mathrm{meas}}$ is the gate count of the final measurement of $O$, one raw execution at $s_{\ell}$ uses
\begin{equation}
    G_{\mathrm{run}}(s_\ell):= g_{\mathrm{step}}\frac{T}{s_\ell}+g_{\mathrm{meas}}
\end{equation}
gates.

\begin{proposition}[Resource Accounting]\label{prop:resourceAccounting}
    Let $N$ be as in \cref{thm:samplingCost}. The maximum single-run circuit depth is 
    \begin{equation}
        D_\mathrm{max}=\max_{1\le\ell\le m} D_\mathrm{run}(s_\ell)=d_\mathrm{step}T\max_{1\le \ell\le m}\frac{1}{s_\ell}.
    \end{equation}
    The total number of raw circuit executions is $N_\mathrm{total}=mN$. The aggregate circuit depth budget, summed over all raw executions, is 
    \begin{equation}
        D_\mathrm{total}=N\sum_{\ell=1}^{m}D_\mathrm{run}(s_\ell)=Nd_\mathrm{step}T\sum_{\ell=1}^{m}\frac{1}{s_\ell}.
    \end{equation}
    The total gate count is
    \begin{equation}
        G_\mathrm{total}=N\sum_{\ell=1}^{m}G_\mathrm{run}(s_\ell)=N\sum_{\ell=1}^{m}\left(g_\mathrm{step}\frac{T}{s_\ell}+g_\mathrm{meas}\right)
    \end{equation}
\end{proposition}
\begin{proof}
    $D_\mathrm{max}$ follows directly from the definition of $D_\mathrm{run}$, since $D_\mathrm{run}(s_\ell)$ decreases as $s_\ell$ increases and the maximum is attained at smallest $s_\ell$. Each step size is executed $N$ times by \cref{thm:samplingCost}, giving $N_\mathrm{total}=mN$ raw executions in total. Since $N$ does not depend on $\ell$, summing the depth $D_\mathrm{run}(s_\ell)$ over all $N$ executions at node $\ell$ and then over all $m$ step sizes gives $D_\mathrm{total}=N\sum_\ell D_\mathrm{run}(s_\ell)$. The same argument applied to the gate count per execution gives $G_\mathrm{total}$.
\end{proof}
Note that $D_\mathrm{total}$ is an aggregate depth budget, the sum of circuit depths over all raw executions, not a gate count. 

\begin{rmk}[Independence from Success-Probability]
    The success probability $p(s_\ell)$ does not appear as a factor in $N$, $D_\mathrm{total}$, or $G_\mathrm{total}$, because failed trajectories are not discarded but contribute the value zero. Consequently, the Hoeffding bound in \cref{thm:samplingCost} depends only on $\|O\|$, not on $p(s_\ell)$. This is specific to estimating $g_O$ using every raw execution. If the target were $f_O(s)=g_O(s)/p(s)$ instead, by discarding failed runs and averaging only over successes, it would require roughly $p(s_\ell)^{-1}$ as many raw executions to obtain the same number of valid samples, reintroducing the dependence on $p(s_\ell)$. 
\end{rmk}

\subsection{Richardson Resources}\label{subsec:richardson}
Richardson extrapolation instantiates \cref{thm:biasBound} and \cref{prop:resourceAccounting} with harmonically spaced step sizes. Let $s_\mathrm{max}>0$ be a prescribed upper bound on the largest step size, and define
\begin{equation}
    R_1=\left\lceil\frac{T}{s_\mathrm{max}}\right\rceil, \quad R_\ell=\ell R_1, \quad s_\ell=\frac{T}{R_\ell}, \quad  \ell=1,\ldots, m.
\end{equation}
Since $R_\ell\in\mathbb{N}$, all of the step sizes are realizable. Moreover,
\begin{equation}
    s_\ell=\frac{s_1}{\ell}, \quad s_1=\frac{T}{\lceil T/s_{\mathrm{max}} \rceil},
\end{equation}
so the realizable nodes retain the harmonic spacing exactly. We assume throughout that $s_\mathrm{max}\le T$. Then,
\begin{equation}
    \frac{s_\mathrm{max}}{2}\le s_1\le s_\mathrm{max}.
\end{equation}
The upper bound is given by the construction. From $\frac{T}{s_\mathrm{max}} \ge \frac{1}{2}\left\lceil\frac{T}{s_\mathrm{max}}\right\rceil$, we get the lower bound $s_1=\frac{T}{\lceil T/s_\mathrm{max}\rceil}\ge \frac{s_\mathrm{max}}{2}$.

\begin{lemma}[Richardson Node Bounds]\label{lem:richardsonNodes}
    Under this condition,
    \begin{equation}
        \prod_{\ell=1}^{m}s_\ell \le \frac{(s_\mathrm{max})^m}{m!}, \quad 
        \max_{1\le \ell \le m}\frac{1}{s_\ell}\le \frac{2m}{s_\mathrm{max}}, \quad 
        \sum_{\ell=1}^{m}\frac{1}{s_\ell}\le \frac{m(m+1)}{s_\mathrm{max}}.
    \end{equation}
\end{lemma}
\begin{proof}
    The equation $s_\ell=\frac{s_1}{\ell}$ and the upper bound $s_1\le s_\mathrm{max}$ gives $\prod_\ell s_\ell= \prod_\ell \frac{s_1}{\ell}=\frac{s_1^m}{m!}\le\frac{s_\mathrm{max}^m}{m!}$. 
    
    The lower bound $s_\ell=\frac{s_1}{\ell} \ge \frac{s_\mathrm{max}}{2\ell}$ gives $\frac{1}{s_\ell}\le \frac{2\ell}{s_\mathrm{max}}$, and taking the maximum over $\ell=1,\ldots, m$ attains its largest value at $\ell=m$. This gives $\max_\ell \frac{1}{s_\ell}\le \frac{2m}{s_\mathrm{max}}$. 
    
    Summing the same bound over $\ell=1,\ldots,m$ gives $\sum_\ell\frac{1}{s_\ell} \le \frac{2}{s_\mathrm{max}}\sum_{\ell=1}^{m}\ell=\frac{m(m+1)}{s_\mathrm{max}}$.
\end{proof}

The product bound controls the deterministic bias, and the two reciprocal bounds control the maximum depth, the aggregate depth, and gate budgets.
\begin{theorem}[Richardson Resources]\label{thm:richardsonResources} 
    Let $Q:=\max\{S,T\Gamma_{\mathrm{loc}}\}$. There is a universal constant $c>0$ such that, choosing $s_\mathrm{max}:=c/Q$, the deterministic bias is at most $\epsilon/2$ provided
    \begin{equation}
        m=\mathcal{O}\left(\log\frac{\|O\|}{\epsilon}\right).
    \end{equation}
    With this choice, the maximum single-run circuit depth is 
    \begin{equation}
        D^\mathrm{Rich}_\mathrm{max}=\mathcal{O}\left(d_\mathrm{step}TQ\log\frac{\|O\|}{\epsilon}\right),
    \end{equation}
    the aggregate depth budget is 
    \begin{equation}
        D^\mathrm{Rich}_\mathrm{total}=\mathcal{O}\left(\frac{\|O\|^2\Lambda_m^2}{\epsilon^2}\log\frac{m}{\delta}\cdot d_\mathrm{step}TQ\left(\log\frac{\|O\|}{\epsilon}\right)^2\right),
    \end{equation}
    and the total gate count is 
    \begin{equation}
        G^\mathrm{Rich}_\mathrm{total}=\mathcal{O}\left(\frac{\|O\|^2\Lambda_m^2}{\epsilon^2}\log\frac{m}{\delta}\cdot g_\mathrm{step}TQ\left(\log\frac{\|O\|}{\epsilon}\right)^2\right).
    \end{equation}
\end{theorem}
\begin{proof}
    The assumption that $s_\mathrm{max}\le T$, equivalently $c\le QT$, holds once $T$ is not too small relative to $S$ or to the dissipation rate $\Gamma_{\mathrm{loc}}$. Since $QT=\max\{ST,T^2\Gamma_\mathrm{loc}\}$, it suffices that $T\ge c/S$ or $T\ge \sqrt{c/\Gamma_\mathrm{loc}}$.
    By \cref{thm:biasBound} and \cref{lem:richardsonNodes}, the deterministic bias satisfies
    \begin{equation}
        \left|\sum_{\ell=1}^{m}a_\ell g_O(s_\ell)-g_O(0)\right| \le C\|O\|(C'Q)^m\prod_{\ell=1}^{m}s_\ell\le C\|O\|(C'Qs_\mathrm{max})^m\frac{1}{m!}=C\|O\|\frac{(C'c)^m}{m!},
    \end{equation}
    using $s_\mathrm{max}=c/Q$. By Stirling's bound $m!\ge (m/e)^m$, 
    \begin{equation}
        \frac{(C'c)^m}{m!}\le \left(\frac{eC'c}{m}\right)^m.
    \end{equation}
    For $m>2eC'c$, the base satisfies $eC'c/m\le 1/2$, so the right hand side is at most $2^{-m}$. The bias is then bounded by $\epsilon/2$ once $2^{-m}\le \epsilon/(2C\|O\|)$ or $m\ge \log_2(2C\|O\|/\epsilon)$. Combining both requirements, $m\ge \max\{2eC'c,\log_2(2C\|O\|/\epsilon)\}$ suffices. Since $2eC'c$ is a fixed constant, this is $m=\mathcal{O}(\log(\|O\|/\epsilon))$.

    By \cref{prop:resourceAccounting} and \cref{lem:richardsonNodes},
    \begin{equation}
        D^\mathrm{Rich}_\mathrm{max}=d_\mathrm{step}T\max_\ell\frac{1}{s_\ell}=\mathcal{O}\left(\frac{d_\mathrm{step}Tm}{s_\mathrm{max}}\right)=\mathcal{O}\left(d_\mathrm{step}TQm\right),
    \end{equation}
    and substituting $m=\mathcal{O}(\log\frac{\|O\|}{\epsilon})$ gives the stated form. Likewise, 
    \begin{equation}
        D^\mathrm{Rich}_\mathrm{total}=Nd_\mathrm{step}T\sum_\ell\frac{1}{s_\ell}=\mathcal{O}\left(\frac{Nd_\mathrm{step}Tm(m+1)}{s_\mathrm{max}}\right)=\mathcal{O}(Nd_\mathrm{step}TQm^2),
    \end{equation}
    and 
    \begin{equation}
        G^\mathrm{Rich}_\mathrm{total}=N\sum_\ell\left(g_\mathrm{step}\frac{T}{s_\ell}+g_\mathrm{meas}\right)=\mathcal{O}(Ng_\mathrm{step}TQm^2+Nmg_\mathrm{meas})=\mathcal{O}(Ng_\mathrm{step}TQm^2),
    \end{equation}
    where $N$ is as in \cref{thm:samplingCost}, and substituting $m=\mathcal{O}(\log\frac{\|O\|}{\epsilon})$ gives the stated forms. $Nmg_\mathrm{meas}$ term was absorbed as $m^2$ term eventually dominates.
\end{proof}

\subsection{Chebyshev Resources}\label{subsec:chebyshev}
Chebyshev interpolation instantiates the same two results with nodes clustered according to Chebyshev polynomial roots rather than spaced harmonically. Let $s_\mathrm{max}>0$, and suppose the nodes are realizable, $s_\ell=T/R_\ell$ for some $R_\ell\in\mathbb{N}$, and satisfy
\begin{equation}
    c_{ch}\theta_\ell \le s_\ell\le \theta_\ell, \quad \theta_\ell:=s_\mathrm{max}\cos^2\left(\frac{(2\ell-1)\pi}{4m}\right), \quad \ell=1,\ldots, m,
\end{equation}
for a universal constant $c_{ch}\in(0,1]$.

\begin{lemma}[Chebyshev Node Bounds]\label{lem:chebyshevNodes}
    Under this condition,
    \begin{equation}
        \prod_{\ell=1}^{m}s_\ell \le \frac{s_\mathrm{max}^m}{2^{2m-1}}, \quad \max_{1\le \ell \le m}\frac{1}{s_\ell}\le \mathcal{O}\left(\frac{m^2}{c_{ch}s_\mathrm{max}}\right), \quad \sum_{\ell=1}^{m}\frac{1}{s_\ell}\le \frac{2m^2}{c_{ch}s_\mathrm{max}}.
    \end{equation}
\end{lemma}
\begin{proof}
    The upper bound $s_\ell\le \theta_\ell$ gives $\prod_\ell s_\ell\le \prod_\ell \theta_\ell=s_\mathrm{max}^m\prod_\ell\cos^2(\frac{(2\ell-1)\pi}{4m})$. By \cref{lem:chebIdentities}, this equals $s_\mathrm{max}^m\cdot 2^{-2m+1}=s_\mathrm{max}^m/2^{2m-1}$. For the reciprocal sum, the lower bound $s_\ell\ge c_{ch}\theta_\ell$ gives $1/s_\ell\le \frac{1}{c_{ch}s_\mathrm{max}}\sec^2(\frac{(2\ell-1)\pi}{4m})$, and \cref{lem:chebIdentities} gives $\sum_\ell1/s_\ell\le 2m^2/(c_{ch}s_\mathrm{max})$. The smallest node occurs at $\ell=m$, where the angle $\frac{(2m-1)\pi}{4m}$ approaches $\pi/2$ as $m$ grows. The small angle behavior of cosine near $\pi/2$ gives $\theta_m=\mathcal{O}(s_\mathrm{max}/m^2)$, and hence $\max_\ell 1/s_\ell=\mathcal{O}(m^2/(c_{ch}s_\mathrm{max}))$.
\end{proof}

The reciprocal sum agrees with Richardson's order in $m$. The maximum reciprocal does not, since the smallest Chebyshev node decreases quadratically in $m$ rather than linearly, yielding a different scaling for the maximum depth.
\begin{theorem}[Chebyshev Resources]\label{thm:chebyshevResources} 
    Let $Q:=\max\{S,T\Gamma_{\mathrm{loc}}\}$. There is a universal constant $c>0$ such that, choosing $s_\mathrm{max}:=c/Q$, the deterministic bias is at most $\epsilon/2$ provided
    \begin{equation}
        m=\mathcal{O}\left(\log\frac{\|O\|}{\epsilon}\right).
    \end{equation}
    With this choice, the maximum single-run circuit depth is 
    \begin{equation}
        D^\mathrm{Cheb}_\mathrm{max}=\mathcal{O}\left(d_\mathrm{step}TQ\left(\log\frac{\|O\|}{\epsilon}\right)^2\right),
    \end{equation}
    the aggregate depth budget is 
    \begin{equation}
        D^\mathrm{Cheb}_\mathrm{total}=\mathcal{O}\left(\frac{\|O\|^2\Lambda_m^2}{\epsilon^2}\log\frac{m}{\delta}\cdot d_\mathrm{step}TQ\left(\log\frac{\|O\|}{\epsilon}\right)^2\right),
    \end{equation}
    and the total gate count is 
    \begin{equation}
        G^\mathrm{Cheb}_\mathrm{total}=\mathcal{O}\left(\frac{\|O\|^2\Lambda_m^2}{\epsilon^2}\log\frac{m}{\delta}\cdot g_\mathrm{step}TQ\left(\log\frac{\|O\|}{\epsilon}\right)^2\right).
    \end{equation}
\end{theorem}
\begin{proof}
    By \cref{thm:biasBound} and \cref{lem:chebyshevNodes}, the deterministic bias satisfies
    \begin{equation}
        \left|\sum_{\ell=1}^{m}a_\ell g_O(s_\ell)-g_O(0)\right| \le C\|O\|(C'Q)^m\prod_{\ell=1}^{m}s_\ell\le C\|O\|(C'Qs_\mathrm{max})^m\frac{1}{2^{2m-1}}=2C\|O\|\left(\frac{C'c}{4}\right)^m,
    \end{equation}
    using $s_\mathrm{max}=c/Q$. Let $r:=C'c/4$, and choose $c\le 2/C'$, so that $r\le 1/2$. The bias is then bounded by $2C\|O\|r^{m}\le \epsilon/2$ once $r^m\le \epsilon/(4C\|O\|)$ or $m\log (1/r)\ge \log (4C\|O\|/\epsilon)$. Since $r\le 1/2$ gives $\log(1/r)\ge \log 2$, it suffices that $m\ge \log_2(4C\|O\|/\epsilon)$, which is $m=\mathcal{O}(\log(\|O\|/\epsilon))$.

    By \cref{prop:resourceAccounting} and \cref{lem:chebyshevNodes},
    \begin{equation}
        D^\mathrm{Cheb}_\mathrm{max}=d_\mathrm{step}T\max_\ell\frac{1}{s_\ell}=\mathcal{O}\left(\frac{d_\mathrm{step}Tm^2}{s_\mathrm{max}}\right)=\mathcal{O}(d_\mathrm{step}TQm^2),
    \end{equation}
    and substituting $m=\mathcal{O}(\log\frac{\|O\|}{\epsilon})$ gives the stated form. Similarly, 
    \begin{equation}
        D^\mathrm{Cheb}_\mathrm{total}=Nd_\mathrm{step}T\sum_\ell\frac{1}{s_\ell}=\mathcal{O}\left(\frac{Nd_\mathrm{step}Tm^2}{s_\mathrm{max}}\right)=\mathcal{O}(Nd_\mathrm{step}TQm^2),
    \end{equation}
    and 
    \begin{equation}
        G^\mathrm{Cheb}_\mathrm{total}=N\sum_\ell\left(g_\mathrm{step}\frac{T}{s_\ell}+g_\mathrm{meas}\right)=\mathcal{O}(Ng_\mathrm{step}TQm^2+Nmg_\mathrm{meas})=\mathcal{O}(Ng_\mathrm{step}TQm^2),
    \end{equation}
    where $N$ is as in \cref{thm:samplingCost}, and substituting $m=\mathcal{O}(\log\frac{\|O\|}{\epsilon})$ gives the stated forms.
\end{proof}

\subsection{Geometrically Local Scaling}\label{subsec:localScaling}
Assume $H$ and each $L_j$ are geometrically local on $n$ qubits. Each acts on $\mathcal{O}(1)$ qubits arranged on a constant dimensional lattice, $J=\mathcal{O}(n)$, all local coefficients are $\mathcal{O}(1)$, and each local term fails to commute with at most $\mathcal{O}(1)$ other local terms.

\begin{lemma}[Geometrically Local Constants]\label{lem:localConst}
    Under these assumptions, $S=\mathcal{O}(n)$ and $\Gamma_\mathrm{loc}=\mathcal{O}(n)$.
\end{lemma}
\begin{proof}
    $S=\|H\|+\sum_j\|D_j\|$ is a sum of $\mathcal{O}(n)$ terms each $\mathcal{O}(1)$, hence $\mathcal{O}(n)$. $\Gamma_\mathrm{PF}=\sum_{\mu<\nu}\|[B_\nu,B_\mu]\|$ has $\mathcal{O}(n)$ nonzero terms, since each $B_\nu$ fails to commute with only $\mathcal{O}(1)$ others, and each nonzero commutator is $\mathcal{O}(1)$ in norm, giving $\Gamma_\mathrm{PF}=\mathcal{O}(n)$. $\Gamma_{\mathrm{cos}}=\frac{2}{3}\sum_j\|D^2_j\|$ is a sum of $\mathcal{O}(n)$ terms each $\mathcal{O}(1)$, hence $\mathcal{O}(n)$. Since $\Gamma_\mathrm{loc}=\frac{1}{2}\Gamma_{\mathrm{PF}}+\Gamma_{\mathrm{cos}}$, $\Gamma_\mathrm{loc}=\mathcal{O}(n)$.
\end{proof}

In a single-ancilla sequential implementation, one time step applies the $\mathcal{O}(n)$ Hamiltonian terms and the $J=\mathcal{O}(n)$ dissipative blocks one at a time, each requiring $\mathcal{O}(1)$ depth and gates by geometric locality, so $d_\mathrm{step}=\mathcal{O}(n)$ and $g_\mathrm{step}=\mathcal{O}(n)$. If $O$ is $k$-local, $g_\mathrm{meas}=\mathcal{O}(k)$.

If $T\Gamma_\mathrm{loc}\ge S$, then $Q=T\Gamma_\mathrm{loc}=\mathcal{O}(nT)$. 

\begin{corollary}[Geometrically Local Maximum Depth]\label{cor:maxDepth}
    In this regime, 
    \begin{equation}
        D^\mathrm{Rich}_\mathrm{max}=\mathcal{O}\left(n^2T^2\log\frac{\|O\|}{\epsilon}\right), \quad D^\mathrm{Cheb}_\mathrm{max}=\mathcal{O}\left(n^2T^2\left(\log\frac{\|O\|}{\epsilon}\right)^2\right).
    \end{equation}
\end{corollary}
\begin{proof}
    Substituting $Q=\mathcal{O}(nT)$ from \cref{lem:localConst} into the maximum depth expressions of \cref{thm:richardsonResources} and \cref{thm:chebyshevResources} gives the desired expression.
\end{proof}

The remaining quantities, $N$, $D_\mathrm{total}$, and $G_\mathrm{total}$, depend additionally on $\Lambda_m$, which has not yet been bounded for either method.

\begin{lemma}[Richardson Weight Bound]\label{lem:richWeights}
    For the ideal and realizable harmonic nodes $s_\ell=\frac{s_1}{\ell}$, where $s_1=\frac{T}{\lceil T/s_\mathrm{max}\rceil}$, the exact extrapolation weights satisfy
    \begin{equation}
        |a_\ell|=\frac{\ell^{m-1}}{(\ell-1)!(m-\ell)!},
    \end{equation}
    and consequently $\Lambda_m^\mathrm{Rich}=\Or((2e)^m/\sqrt{m})$.
\end{lemma}
\begin{proof}
    Substituting $s_k=s_1/k$ into $a_\ell=\prod_{k\ne \ell}\frac{-s_k}{s_\ell-s_k}$ gives $\frac{-s_k}{s_\ell-s_k}=\frac{\ell}{\ell-k}$ for each $k\ne \ell$, so
    \begin{equation}
        |a_\ell|=\left|\prod_{k=1,k\ne \ell}^{m}\frac{\ell}{\ell-k}\right|=\ell^{m-1}\prod_{k\ne \ell}\frac{1}{|\ell-k|}=\frac{\ell^{m-1}}{(\ell-1)!(m-\ell)!},
    \end{equation}
    using $\prod_{k=1}^{\ell}(\ell-k)=(\ell-1)!$ and $\prod_{k=\ell+1}^{m}(k-\ell)=(m-\ell)!$. For the upper bound, since $\ell\le m$, we have
    \begin{equation}
        \Lambda_m^\mathrm{Rich}
        \le
        m^{m-1}\sum_{\ell=1}^{m}  \frac{1}{(\ell-1)!(m-\ell)!}
        =
        \frac{m^{m-1}}{(m-1)!} \sum_{k=0}^{m-1}\binom{m-1}{k}
        =
        m^{m-1} \frac{2^{m-1}}{(m-1)!} 
        =
        \Or\left(\frac{(2e)^m}{\sqrt{m}} \right).
    \end{equation}
    Here in the last line, we used the Stirling's formula, 
    \begin{equation}
        (m-1)!\sim \sqrt{2\pi(m-1)}\left(\frac{m-1}{e}\right)^{m-1},
    \end{equation}
    together with the fact that $(m/(m-1))^{m-1}\rightarrow e$.
\end{proof}

Standard Chebyshev interpolation theory gives a Lebesgue constant bound of $\Or(\log m)$ for the ideal nodes ~\cite[Lemma~8]{MohammadipourLi2026}. In an actual digital quantum simulation, the step size must be realizable,
$s_\ell=T/R_\ell$ with $R_\ell\in\mathbb N$. Thus, the implemented nodes are
rounded versions of the ideal nodes, i.e.,
\begin{equation}\label{eq:roundCheyNodes}
    R_\ell
    =
    \left\lceil
    \frac{T}{\theta_\ell}
    \right\rceil,
    \qquad
    s_\ell
    =
    \frac{T}{R_\ell}.
\end{equation}
The effect of this rounding on the weight norm is a purely classical stability question for perturbed Chebyshev nodes. This rounding issue has been carefully studied in the Hamiltonian and Lindblad step-size extrapolation settings~\cite{WatsonWatkins2024,MohammadipourLi2025}. A particularly useful lemma is \cite[Lemma~9]{MohammadipourLi2025}, which we adapt to the present notation below.

\begin{lemma}[Rounded Chebyshev Weight Bound, \cite{MohammadipourLi2025}]
\label{lem:roundedChebWeights}
Let the rounded Chebyshev nodes be defined by \cref{eq:roundCheyNodes}. 
Let
$\Lambda_m^\mathrm{Cheb}$ be the sum of the absolute values of the Lagrange weights,
computed from the actual rounded nodes, as in~\cref{eq:def_a_l_Lambda_m}.
There exist universal constants $b_1, b_2,b_3$ such that if
\begin{equation}
    T\ge b_1s_\mathrm{max}m^2,
\end{equation}
then the rounded nodes are distinct and
\begin{equation}
    \Lambda_m^\mathrm{Cheb}
    \le
    b_2m^{4/(\pi^2-4)}\log m.
\end{equation}
If the stronger resolution condition $T\ge b_3s_\mathrm{max}m^2\log m$ holds, then $\Lambda_m^\mathrm{Cheb}=\Or(\log m)$.
\end{lemma}
Assume throughout that $T\ge b_3{s_\mathrm{max}}m^2\log m$ so that $\Lambda_m^{\mathrm{Cheb}}=\Or(\log m)$.

The growth rate in \cref{lem:richWeights} is qualitatively different from \cref{lem:roundedChebWeights}. Chebyshev's weights grow logarithmically in $m$, while Richardson's grow exponentially. 
\begin{corollary}[Weight Bounds]\label{cor:weightBounds}
    Substituting $m=\mathcal{O}(\log(\|O\|/\epsilon))$ into each gives
    \begin{equation}
       \Lambda_m^\mathrm{Rich}=\left(\frac{\|O\|}{\epsilon}\right)^{\Or(1)} ,\quad
       \Lambda_m^\mathrm{Cheb}=\mathcal{O}\left(\log\log\frac{\|O\|}{\epsilon}\right),
    \end{equation}
    The former is polynomial in $\|O\|/\epsilon$, and the latter is polylogarithmic in $\|O\|/\epsilon$.
\end{corollary}

These weight bounds yield bounds on $N, G_\mathrm{total}$ in terms of the given constants for both extrapolation methods below. Since $d_\mathrm{step}=\Or(n)=g_\mathrm{step}$ in this geometrically local regime, as established above, $D_\mathrm{total}$ is asymptotically equivalent to $G_\mathrm{total}$ and carries no additional information. We therefore report only $G_\mathrm{total}$, which has a direct operational meaning as a total gate count, unlike the depth budget quantity $D_\mathrm{total}$. 

\begin{corollary}[Final Richardson Resource Scaling]\label{cor:richFinal}
    In the geometrically local, large-time regime, choosing $m=\mathcal{O}(\log(\|O\|/\epsilon))$ for Richardson extrapolation, 
    \begin{equation}
        N^\mathrm{Rich}= \left(\frac{\|O\|}{\epsilon}\right)^{\Or(1)}\log\frac{\log(\|O\|/\epsilon)}{\delta},
    \end{equation}
    \begin{equation}
        G^\mathrm{Rich}_\mathrm{total}=\left(\frac{\|O\|}{\epsilon}\right)^{\Or(1)} n^2T^2\log\frac{\|O\|}{\epsilon}\log\frac{\log(\|O\|/\epsilon)}{\delta}.
    \end{equation}
\end{corollary}
\begin{proof}
    Substitute $\Lambda_m^\mathrm{Rich}=(\|O\|/\epsilon)^{\Or(1)}$ into \cref{thm:samplingCost} to obtain $N^\mathrm{Rich}$, and into \cref{thm:richardsonResources} together with $Q=\mathcal{O}(nT)$ to obtain $G^\mathrm{Rich}_\mathrm{total}$.
\end{proof}
In particular, since $N^\mathrm{Rich}$ and $G_\mathrm{total}^\mathrm{Rich}$ both inherit the factor of $\|O\|^2/\epsilon^2$ from Hoeffding concentration (\cref{thm:samplingCost}), in addition to the $\Or(1)$ polynomial contribution from $(\Lambda_m^\mathrm{Rich})^2$ bounded in \cref{cor:weightBounds}, the overall exponent governing the dependence on $\|O\|/\epsilon$ is at least $2$. 

\begin{corollary}[Final Chebyshev Resource Scaling]\label{cor:chebFinal}
    Under the same regime, suppose in addition that the resolution condition $T\ge b_3s_\mathrm{max} m^2\log m$ of \cref{lem:roundedChebWeights} holds. Then choosing $m=\mathcal{O}(\log(\|O\|/\epsilon))$ with rounded Chebyshev nodes,
    \begin{equation}
        N^\mathrm{Cheb}=\mathcal{O}\left(\frac{\|O\|^2}{\epsilon^2}\left(\log\log\frac{\|O\|}{\epsilon}\right)^2\log\frac{\log(\|O\|/\epsilon)}{\delta}\right),
    \end{equation}
    \begin{equation}
        G^\mathrm{Cheb}_\mathrm{total}= \Or\left(\frac{\|O\|^2}{\epsilon^2}\left(\log\log\frac{\|O\|}{\epsilon}\right)^2\log\frac{\log(\|O\|/\epsilon)}{\delta}\cdot n^2T^2\left(\log\frac{\|O\|}{\epsilon}\right)^2\right).
    \end{equation}
\end{corollary}
\begin{proof}
    Substitute $\Lambda_m^\mathrm{Cheb}=\mathcal{O}(\log\log(\|O\|/\epsilon))$ into \cref{thm:samplingCost} to obtain $N^\mathrm{Cheb}$, and into \cref{thm:chebyshevResources} together with $Q=\mathcal{O}(nT)$ to obtain $G^\mathrm{Cheb}_\mathrm{total}$.
\end{proof}

\begin{rmk}[Effects of a Weaker Condition on $T$]
    Using the known bounds and equations $s_\mathrm{max}=c/Q$, $m=\Or(\log(\|O\|/\epsilon))$, $Q=\Or(nT)$, the assumed condition $T\ge b_3s_\mathrm{max} m^2\log m$ can be rewritten as
    \begin{equation}
        T\gtrsim \frac{1}{\sqrt{n}}\log\frac{\|O\|}{\epsilon}\sqrt{\log\log\frac{\|O\|}{\epsilon}}.
    \end{equation}
    If $T$ only satisfies the weaker condition $T\ge b_1s_\mathrm{max}m^2$ or $T\gtrsim\frac{1}{\sqrt{n}}\log\frac{\|O\|}{\epsilon}$, then by \cref{lem:roundedChebWeights}
    \begin{equation}
        \Lambda_m^\mathrm{Cheb} = \Or\left(\left(\log\frac{\|O\|}{\epsilon}\right)^{4/(\pi^2-4)}\log\log\frac{\|O\|}{\epsilon}\right)
    \end{equation}
    rather than $\Or(\log \log(\|O\|/\epsilon))$. Since $N^\mathrm{Cheb}$, $D_\mathrm{total}^\mathrm{Cheb}$, and $G_\mathrm{total}^\mathrm{Cheb}$ each scale with $(\Lambda_m^\mathrm{Cheb})^2$, this multiplies these three quantities by a factor of $\left(\log\frac{\|O\|}{\epsilon}\right)^{8/(\pi^2-4)}$ relative to the strong condition case described in \cref{cor:chebFinal}.
\end{rmk}

\begin{rmk}[Comparison of the Two Methods]
    Only the maximum single circuit depth $D_\mathrm{max}$ is reduced to polylogarithmic scaling by either method. Both methods inherit the $\Theta(1/\epsilon^2)$ sampling cost of Hoeffding concentration, present regardless of the extrapolation scheme used. The two methods differ in whether this baseline cost is multiplied by an additional polylogarithmic factor, as for Chebyshev, or an additional polynomial factor, as for Richardson, the latter arising from the exponential growth of the Richardson weights established in \cref{lem:richWeights}.
\end{rmk}

\subsection{Sampling Lower Bound}\label{subsec:lowerBound}
In the independent-shot sampling model considered here, the total number of raw circuit executions is necessarily $\Omega(1/\epsilon^2)$ for constant failure probability. Since each raw execution has a nonzero gate cost, the total gate count also inherits this $\Omega(1/\epsilon^2)$ sampling factor. This means that the factor of $1/\epsilon^2$ in $N$, $D_\mathrm{total}$, and $G_\mathrm{total}$ stated above is unavoidable for any estimator in the independent-shot sampling model. We note that the lower bound proof is rather standard, which we present here for completeness.

\begin{theorem}[Sampling Lower Bound] \label{thm:samplingLowerBound}
Consider the independent-shot sampling model in which each raw circuit execution produces one classical measurement outcome and all subsequent processing is classical. Any estimator that approximates the unnormalized observable $g_O(0)$ to additive error $\epsilon$ with failure probability at most $\delta$ for all problem instances must use
\begin{equation}
    N = \Omega\left( \frac{\|O\|^2}{\epsilon^2} \log\frac{1}{\delta} \right)
\end{equation}
raw circuit executions in the worst case, for $0<\epsilon\le c\|O\|$ and $0<\delta\le 1/8$.
\end{theorem}
\begin{proof}
It suffices to consider the trivial ODE instance $A=0$, for which the exact and finite-step evolutions are both the identity, and there is no postselection failure. Let
\begin{equation}
    O=B\ket{1}\bra{1}, \quad
    B=\|O\|.
\end{equation}
Consider the two input states
\begin{equation}
    \ket{\psi_\pm} = \sqrt{1-p_\pm}\ket{0}+\sqrt{p_\pm}\ket{1}, \quad
    p_\pm = \frac{1}{2} \pm \frac{2\epsilon}{B}.
\end{equation}
For $\epsilon\le B/8$, both $p_+$ and $p_-$ lie in $(0,1)$. For these two instances,
\begin{equation}
    g_\pm = \bra{\psi_\pm}O\ket{\psi_\pm} = Bp_\pm,
\end{equation}
and hence $|g_{+}-g_{-}|=4\epsilon$. A raw measurement of $O$ produces the value $B$ with probability $p_\pm$ and the value $0$ otherwise. Thus estimating $g_\pm$ is equivalent to estimating the mean of a Bernoulli random variable, scaled by $B$.

Suppose that an estimator using $N$ independent raw samples achieves additive error at most $\epsilon$ with probability at least $1-\delta$ for both instances. Consider the test $\phi$ that outputs $+$ if $\hat g>B/2$ and outputs $-$ otherwise. This test distinguishes the two hypotheses with error probability at most $\delta$ under each hypothesis. Indeed, under the $+$ instance, $g_+=\frac{B}{2}+2\epsilon$, so $|\hat g-g_+|\le \epsilon$ implies $\hat g>B/2$. Similarly, under the $-$ instance, $g_-=\frac{B}{2}-2\epsilon$, so $|\hat g-g_-|\le \epsilon$ implies $\hat g<B/2$.

Let $P_\pm=\mathrm{Bern}(p_\pm)$ denote the Bernoulli distribution with success probability $p_\pm$, $P_\pm^{\otimes N}$ denote the $N$-fold product distribution under $N$ independent samples, and $D_{\mathrm{KL}}(P_1\|P_2)$ denote the Kullback-Leibler divergence. Let $E$ be the event that the test $\phi$ outputs $+$. Since the estimator is correct with probability at least $1-\delta$ under each instance, the probability that $\phi$ errs under the $-$ instance plus the probability that $\phi$ errs under the $+$ instance is at most $2\delta$
\begin{equation}
    P_-^{\otimes N}(\phi=+)+P_+^{\otimes N}(\phi=-)\le 2\delta.
\end{equation}
On the other hand,
\begin{equation}
    P_{-}^{\otimes N}(\phi=+)+P_+^{\otimes N}(\phi=-)=P_-^{\otimes N}(E)+P_+^{\otimes N}(E^\complement)\ge \int\!\min\left\{dP_{-}^{\otimes N},dP_{+}^{\otimes N}\right\}.
\end{equation}
By the Bretagnolle-Huber testing inequality~\cite[Lemma~2.6]{Tsybakov2008},
\begin{equation}
    \int\min\!\left\{dP_{-}^{\otimes N},dP_{+}^{\otimes N}\right\}\ge \frac{1}{2}\exp(-D_\mathrm{KL}(P_-^{\otimes N}\,\|\,P_+^{\otimes N})).
\end{equation}
Since the KL divergence tensorizes over independent samples,
\begin{equation}
    D_{\mathrm{KL}}\left(P_-^{\otimes N}\,\|\,P_+^{\otimes N}\right)
    = N D_{\mathrm{KL}}\left(P_-\,\|\,P_+\right)
    = N D_{\mathrm{KL}}\left(\mathrm{Bern}(p_-)\,\|\,\mathrm{Bern}(p_+)\right).
\end{equation}
For $p_\pm=1/2\pm 2\epsilon/B$ and $\epsilon\le B/8$, the Bernoulli KL divergence obeys
\begin{equation}
    D_{\mathrm{KL}}\left(\mathrm{Bern}(p_-)\,\|\,\mathrm{Bern}(p_+)\right) \le C\frac{\epsilon^2}{B^2}.
\end{equation}
for a universal constant $C$. Combining the four displayed bounds above, we get
\begin{equation}
    2\delta \ge \frac{1}{2}\exp\left(-CN\frac{\epsilon^2}{B^2}\right)
\end{equation}
which rearranges to
\begin{equation}
    N \ge C'\frac{B^2}{\epsilon^2}\log\frac{1}{\delta}
\end{equation}
for a universal constant $C'>0$, using $0<\delta\le 1/8$. Since $B=\|O\|$, this proves the claim.
\end{proof}

\begin{rmk}
The lower bound in \cref{thm:samplingLowerBound} applies to the independent-shot sampling model considered in this work. It does not rule out coherent amplitude-estimation-type procedures, which may improve the dependence on $\epsilon$ at the cost of additional coherent control and more ancilla qubits. Our goal here is to keep the circuit near-term and locality-preserving, using only \textit{one} ancilla qubit. Since our circuit is not fully coherent, including a mid-circuit measurement by design, this tradeoff is not available to it. The result therefore shows that the $\epsilon^{-2}$ factor in the total number of raw circuit executions is optimal for the sampling model used here, while the maximum single-run circuit depth can be reduced through step-size postprocessing.
\end{rmk}

\section{Numerical Results}\label{sec:numerics}
In this numerical section, we provide numerical results supporting our theoretical findings. We test on two differential equation examples. One is the Hatano-Nelson model, an ordinary differential equation, and the other is the convection-diffusion equation, a partial differential equation. 

\subsection{Hatano-Nelson Model}\label{subsec:HatanoNelson}
We consider the Hatano-Nelson model, an important non-Hermitian lattice model in physics~\cite{NelsonHatanoPRL1996, NelsonHatanoPRL1997, NelsonHatanoPRL1998}. The defining feature of this model is non-reciprocal hopping, where the hopping amplitude from site $j$ to site $j+1$ differs from the amplitude in the reverse direction. This feature gives rise to characteristic non-Hermitian phenomena, most notably the non-Hermitian skin effect, in which eigenstates and dynamical profiles accumulate near one boundary under open boundary conditions.

In the numerical experiments, we use the nearest-neighbor interacting spinless fermionic Hatano-Nelson model
\begin{align}
    H_{\rm HN}
    =
    \sum_{j=1}^{N-1}
    \left[
    J(c^{\dagger}_{j+1}c_j+c^{\dagger}_{j}c_{j+1})
    +
    \gamma(c^{\dagger}_{j+1}c_j-c^{\dagger}_{j}c_{j+1})
    +
    V n_j n_{j+1}
    \right],
    \label{eq:numerics_hn_interacting}
\end{align}
where $c_j$ and $c_j^\dagger$ are fermionic annihilation and creation operators and $n_j=c_j^\dagger c_j$. The parameter $J$ gives the reciprocal hopping strength, $\gamma$ controls the non-reciprocal part, and $V$ is the nearest-neighbor interaction strength. This model is also considered in~\cite{FangGeorgeTong2025}, where the Jordan-Wigner transformation and the resulting circuit-level implementation of the one-ancilla ODE algorithm are described in detail. We do not repeat the circuit synthesis here.

In the numerical tests, we use the parameter values
\begin{align}
    N=5,
    \qquad
    J=1,
    \qquad
    \gamma=0.8,
    \qquad
    V=1.5,
    \qquad
    T=2.
    \label{eq:numerics_hn_parameters}
\end{align}
The initial state is the fixed-particle-number computational basis state
\begin{equation}
    \ket{\psi_0}=\ket{01010}.
    \label{eq:numerics_hn_initial_state}
\end{equation}
We consider the edge-imbalance observable
\begin{align}
O_{\rm edge}
=
\sum_{i=1}^N
\frac{N+1-2i}{N-1} n_i,
\qquad
n_i=\frac{I-Z_i}{2}.
\label{eq:edge_imbalance}
\end{align}
This observable measures the signed imbalance of the particle density between the two ends of the chain. The coefficient $(N+1-2i)/(N-1)$ is positive near the left boundary and negative near the right boundary, so a positive value indicates stronger occupation near the left edge, while a negative value indicates stronger occupation near the right edge. It is therefore a compact scalar diagnostic for the boundary accumulation associated with the non-Hermitian skin effect.
We note that the observable is also simple from the measurement perspective. It is a weighted sum of the single-site number operators $n_i$, and is diagonal in the computational basis after the Jordan-Wigner transformation. Thus it can be estimated from site occupation measurements without measuring multi-qubit off-diagonal Pauli strings.

\begin{figure}[pos=htbp]
    \centering
    \begin{minipage}[t]{0.48\linewidth}
        \centering
        \includegraphics[width=\linewidth]{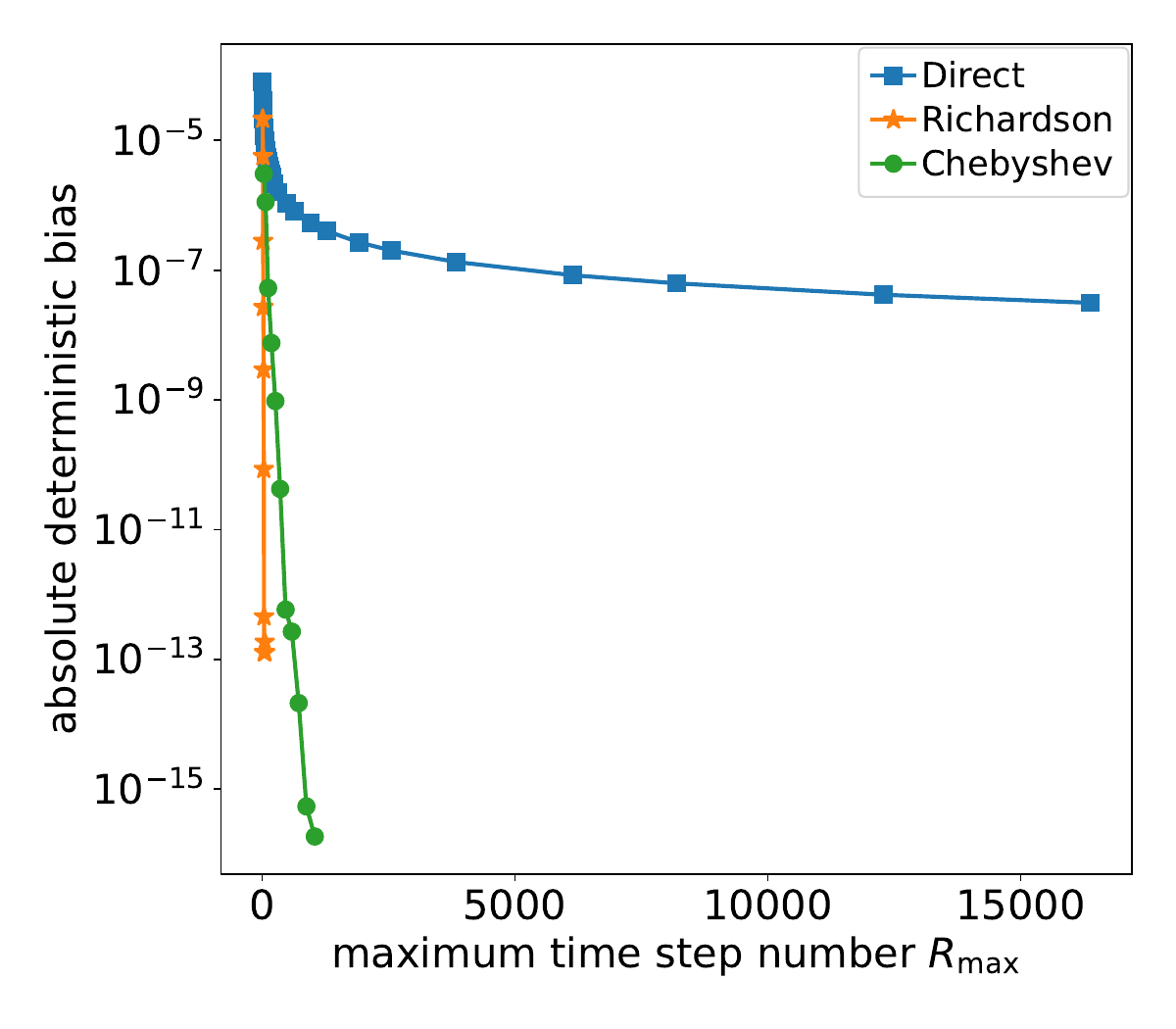}

        \smallskip
        \textbf{(a)} Bias versus $R_{\max}$.
    \end{minipage}
    \hfill
    \begin{minipage}[t]{0.48\linewidth}
        \centering
        \includegraphics[width=\linewidth]{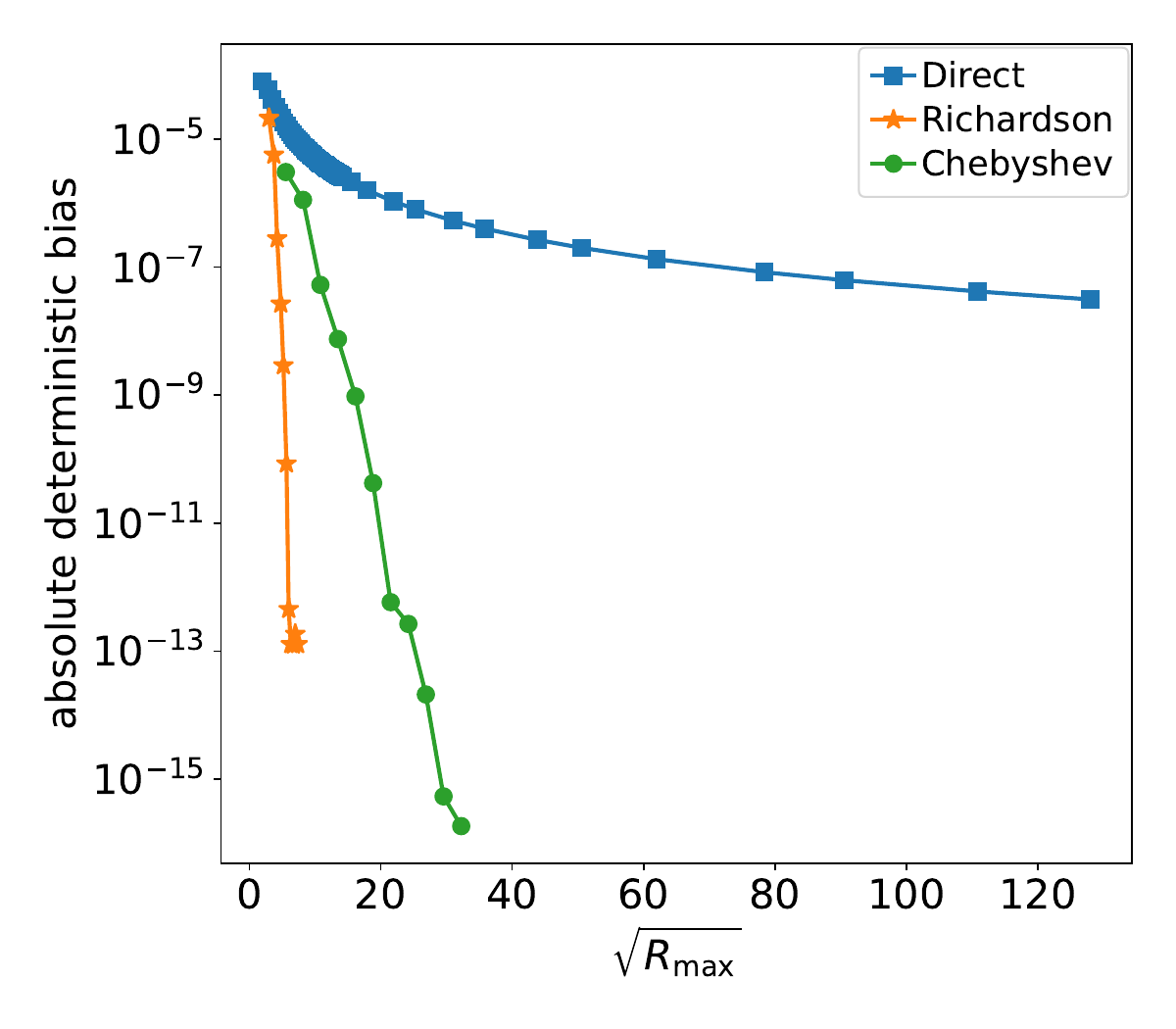}

        \smallskip
        \textbf{(b)} Bias versus $\sqrt{R_{\max}}$.
    \end{minipage}
    \caption{Hatano-Nelson model with edge observable. Both panels are semilogy plots: the y-axis is the absolute deterministic bias on a logarithmic scale, while the x-axis is linear. Panel (a) plots the data against $R_{\max}$, and panel (b) plots the same data against $\sqrt{R_{\max}}$. The direct method without the proposed postprocessing requires significantly larger circuit depth, while Richardson extrapolation and Chebyshev reduce the bias much faster.}
    \label{fig:hn-error}
\end{figure}

\cref{fig:hn-error} shows the deterministic-bias reduction at fixed data points. The $x$ axis is $R_{\max}=\max_j R_j$, the maximum number of time steps among the finite-step circuits used in the extrapolation, where $s_j=T/R_j$. The actual maximum single-run circuit depth is
\begin{equation}
D_{\max}=d_{\rm step}R_{\max},
\end{equation}
up to lower-order measurement overheads. The $y$ axis shows the absolute deterministic bias
\begin{equation}
   \left|g_O^{\rm approx}-g_O(0)\right|, 
\end{equation}
where
$g_O(0)=\bra{\psi_0}(e^{TA})^{\dagger}O(e^{TA})\ket{\psi_0}$ is the exact
zero-step reference and $g_O^{\rm approx}$ is either the direct finite-step value or the extrapolated value. 
The direct method behaves as expected for a first-order discretization: decreasing the step size improves the error slowly. Richardson and Chebyshev both reduce the bias rapidly as the number of nodes increases. This supports the derivative-bound picture that $g_O(s)$ is smooth enough near $s=0$ for polynomial extrapolation to cancel or interpolate the leading step-size error terms.

\begin{figure}[pos=htbp]
    \centering
    \includegraphics[width=0.58\linewidth]{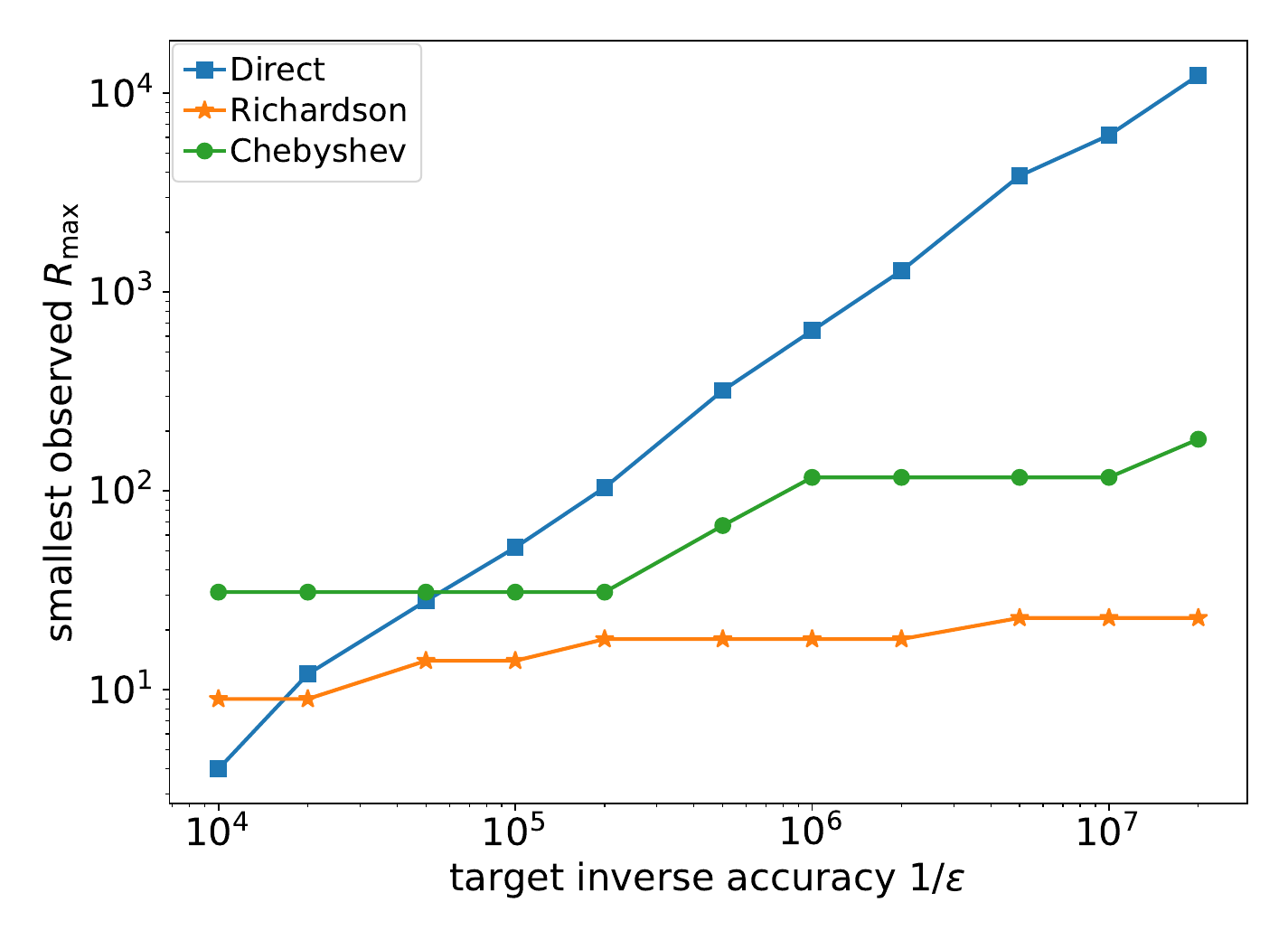}
    \caption{Hatano-Nelson edge observable. This is a log-log plot: the x-axis is the target inverse accuracy $1/\epsilon$, and the y-axis is the smallest observed maximum step number $R_{\max}$ needed to reach the target deterministic bias. The direct method follows the expected inverse-accuracy trend, while the extrapolated methods require much smaller maximum depths over the tested accuracy range. The final upward step in the Chebyshev curve is a discrete degree-selection effect.}
    \label{fig:hn-required-depth}
\end{figure}

At a higher degree, the Richardson curve can stop improving. This is expected in finite precision. If the computed scalar data are perturbed as
$
    \widetilde g(s_j)=g(s_j)+\xi_j,
$
then the extrapolated perturbation is
$
    \sum_j a_j\xi_j,
$
with bound
\begin{equation}
    \left|\sum_j a_j\xi_j\right|\le \Lambda_m\max_j |\xi_j|,
    \qquad
    \Lambda_m=\sum_j |a_j|.
\end{equation}
Richardson weights become increasingly oscillatory as $m$ grows, so a small numerical perturbation can be amplified. Chebyshev interpolation has a milder weight growth in these tests and gives smoother high-degree behavior. This is why Chebyshev is the more stable choice when finite precision, node rounding, or finite-shot noise is present.

\cref{fig:hn-required-depth} inverts the previous comparison. For each target accuracy $\epsilon$, we record the smallest observed $R_{\max}$ such that this precision is reached.
The direct curve grows quickly as $1/\epsilon$ increases, consistent with first-order accuracy. The Richardson and Chebyshev curves do not change much. In particular, they are step-like because both $m$ and $R$ are integer-valued, but they remain far below the direct curve in the high-accuracy regime. 
The last visible upward step in the Chebyshev curve is a degree-selection effect. In the Hatano-Nelson data, the $m=4$ Chebyshev interpolant has
\begin{equation}
    R_{\max}=117,
    \qquad
    E\approx 5.30\times 10^{-8}.
\end{equation}
It satisfies all targets down to $\epsilon=10^{-7}$, but it does not satisfy the stricter target $\epsilon=5\times 10^{-8}$. The selected degree then increases to $m=5$, giving
\begin{equation}
    R_{\max}=182,
    \qquad
    E\approx 7.53\times 10^{-9}.
\end{equation}
The jump $117\to 182$ is therefore not an instability. It reflects the discrete change in $m$ and the relation $R_{\max}=\Theta(m^2)$.

\subsection{Convection-Diffusion Equation}\label{subsec:convenctionDiffusion}

The second benchmark is a two-dimensional periodic convection-diffusion equation,
\begin{equation}
    \partial_t c(x,y,t)=-\nabla\cdot\left(v(x,y)c(x,y,t)\right)+\nu\Delta c(x,y,t).
\end{equation}
In the numerical example, we choose the following parameters
\begin{equation}
    v(x,y)=\alpha(\sin y,\sin x), \quad \alpha =1.2, \quad \quad \nu=\frac{1}{2}, \quad T=0.75.
\end{equation}
 The equation is discretized on a $6\times 6$ Fourier spectral grid. We note that this small-sized discretization is only meant for demonstration purposes. The initial condition is a normalized Gaussian on the torus, namely,
\begin{equation}
c_0(x,y)
=
C\exp\left(
-\frac{
d_{\rm torus}(x,\pi/2)^2+d_{\rm torus}(y,\pi/2)^2
}{2\sigma^2}
\right),
\qquad
\sigma=0.55,
\end{equation}
where $C$ is chosen so that $\|c_0\|_2=1$, and
\begin{equation}
d_{\rm torus}(a,b)
=
\min\{|a-b|,\ 2\pi-|a-b|\}.
\end{equation}
The observable is the low-Fourier-shell projector
\begin{equation}
    O_{\rm low}=\sum_{0<|k|\le 2}|k\rangle\langle k|.
\end{equation}
This observable measures the unnormalized low-frequency content of the solution.
We perform similar numerical tests to the Hatano-Nelson model for the convection-diffusion equation.

\begin{figure}[pos=htbp]
    \centering
    \begin{minipage}[t]{0.48\linewidth}
        \centering
        \includegraphics[width=\linewidth]{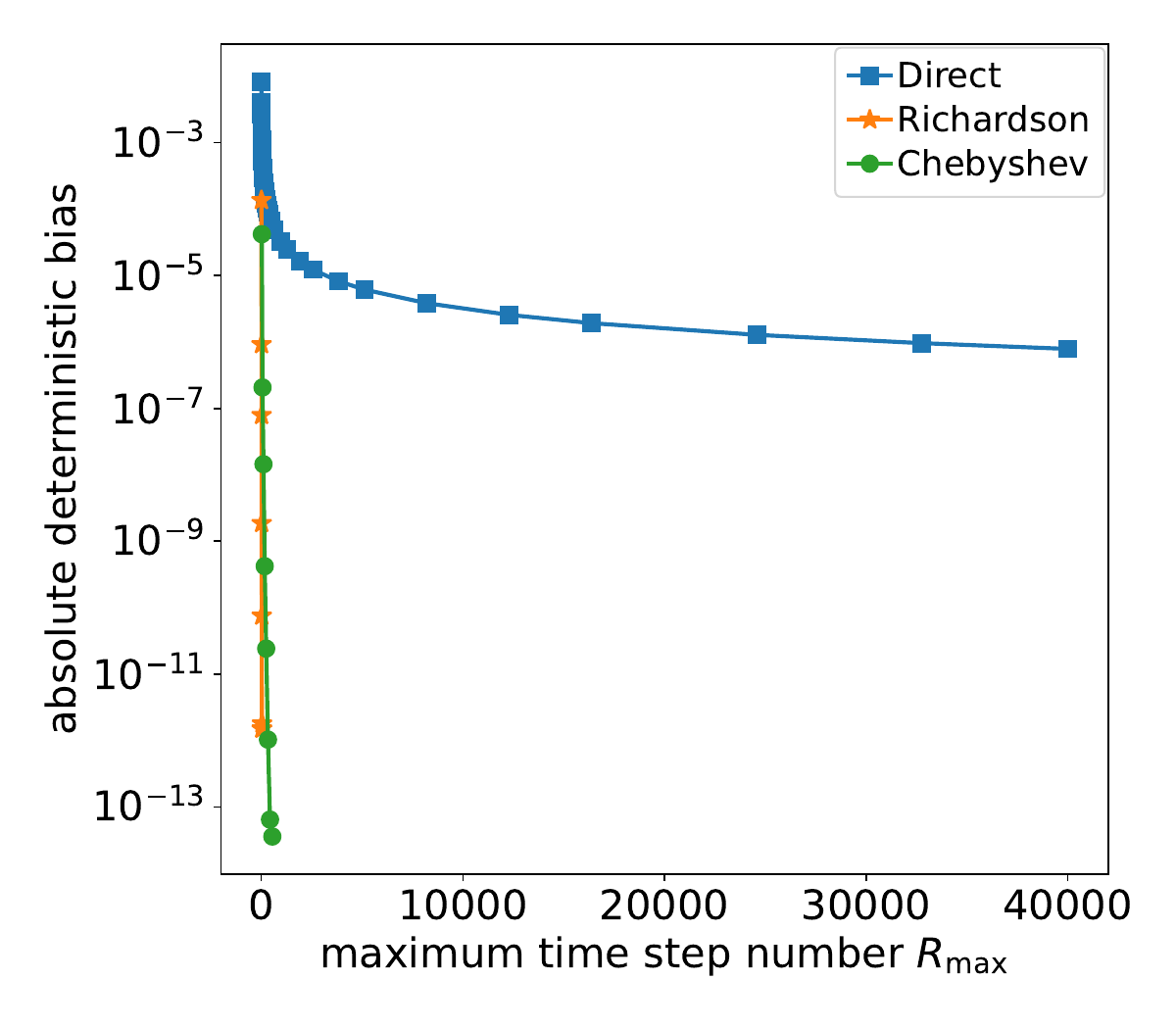}

        \smallskip
        \textbf{(a)} Bias versus $R_{\max}$.
    \end{minipage}
    \hfill
    \begin{minipage}[t]{0.48\linewidth}
        \centering
        \includegraphics[width=\linewidth]{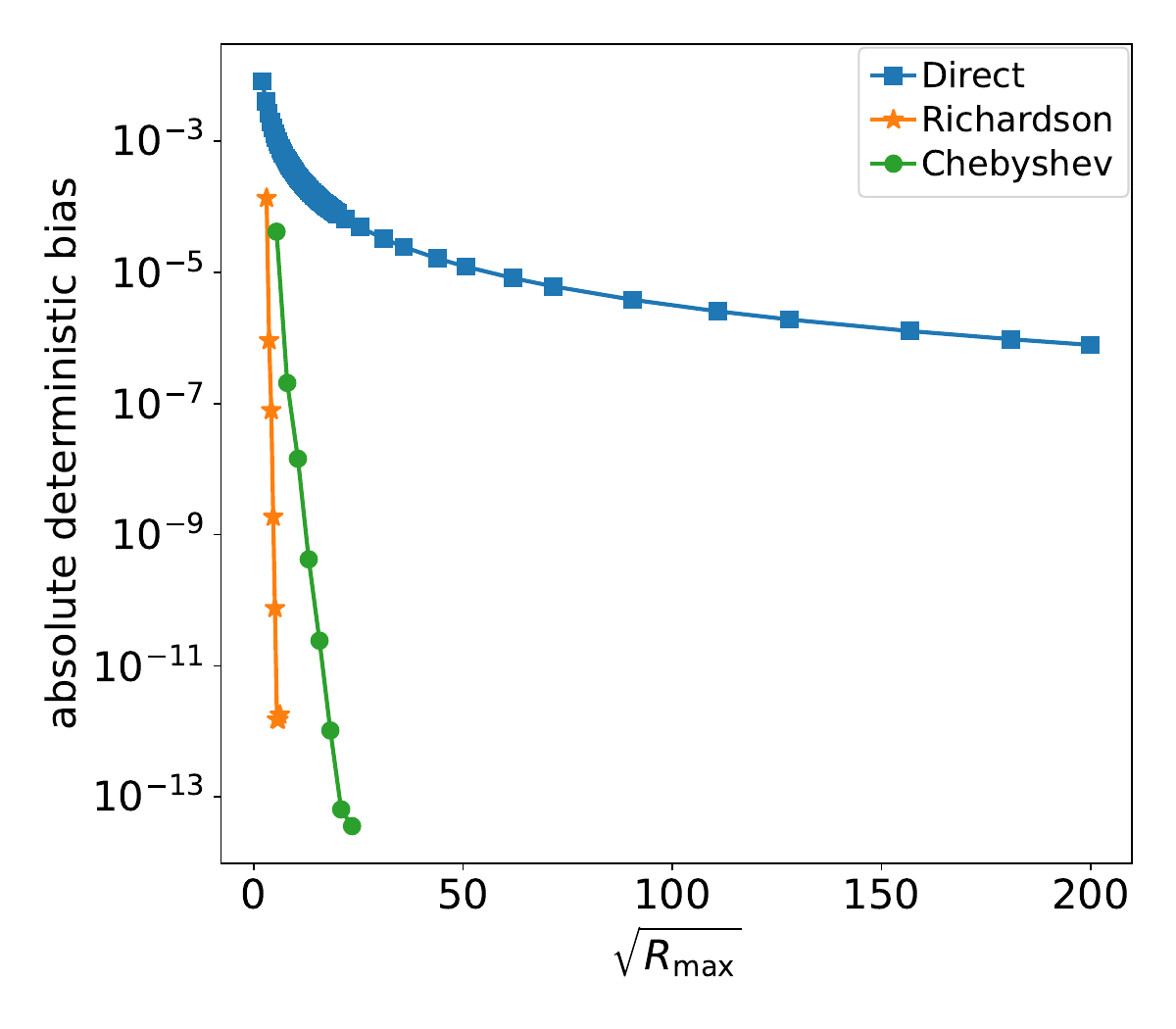}

        \smallskip
        \textbf{(b)} Bias versus $\sqrt{R_{\max}}$.
    \end{minipage}
    \caption{Convection-diffusion low-Fourier-shell observable. Both panels are semilogy plots: the y-axis is the absolute deterministic bias on a logarithmic scale, while the x-axis is linear. Panel (a) plots the data against $R_{\max}$, and panel (b) plots the same data against $\sqrt{R_{\max}}$. For Richardson and Chebyshev, the plotted degrees are $m=2,\ldots,9$. The extrapolated curves show rapid deterministic bias reduction relative to direct first-order simulation.}
    \label{fig:cd-error}
\end{figure}
\begin{figure}[pos=htbp]
    \centering
    \includegraphics[width=0.58\linewidth]{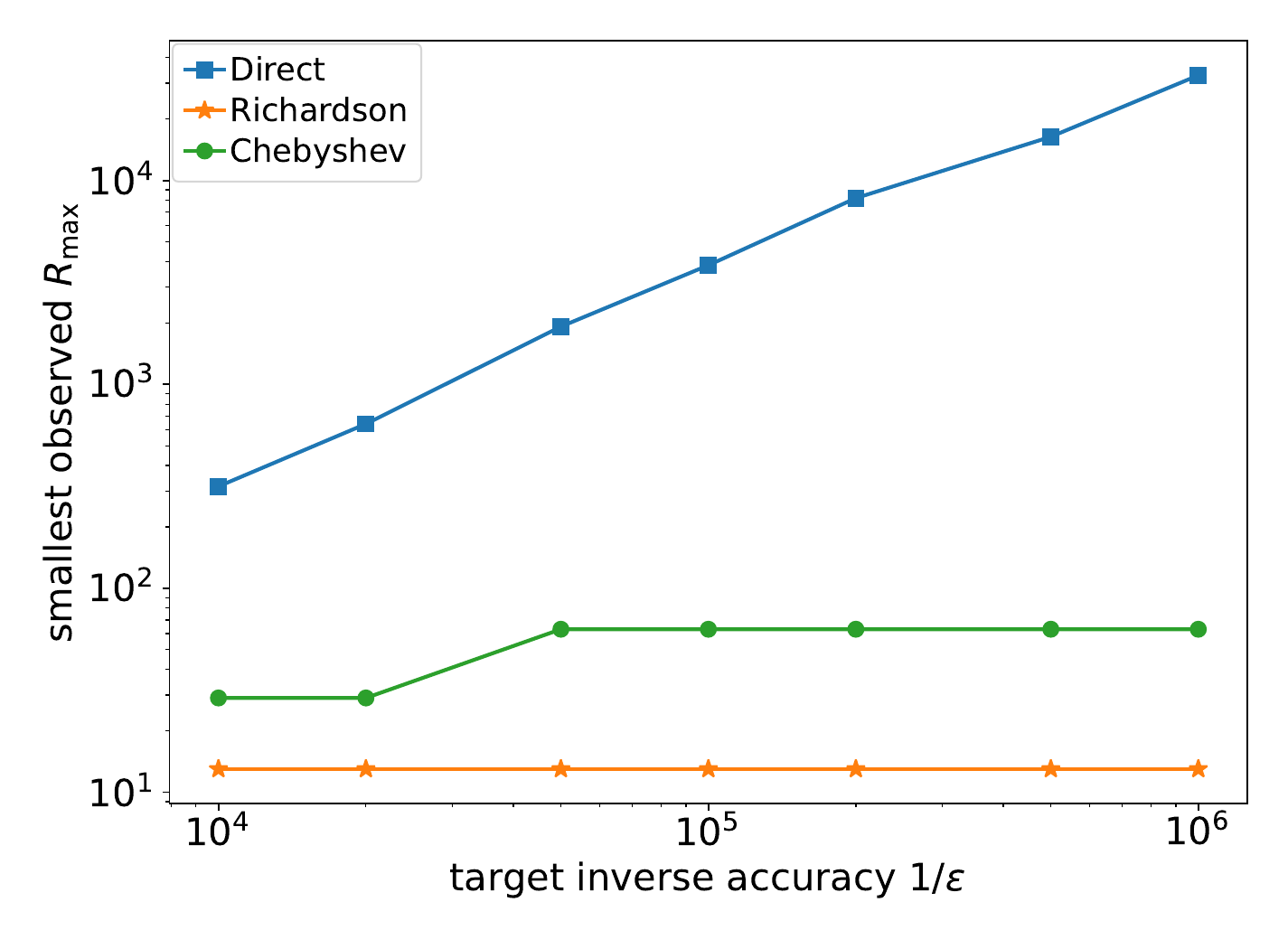}
    \caption{Convection-diffusion low-Fourier-shell observable. This is a log-log plot: the x-axis is the target inverse accuracy $1/\epsilon$, and the y-axis is the smallest observed maximum step number $R_{\max}$ needed to reach the target deterministic bias. The direct method requires rapidly increasing $R_{\max}$ as the target accuracy is tightened, while Richardson extrapolation and Chebyshev interpolation remain much lower over the tested range.}
    \label{fig:cd-required-depth}
\end{figure}

\cref{fig:cd-error} shows that the extrapolation mechanism is consistent with the theory. The convection-diffusion observable also exhibits rapid bias reduction under Richardson and Chebyshev postprocessing. The plotted extrapolated curves use $m=2,\ldots,9$. This range shows the deterministic decay regime before the high-degree numerical floor dominates the experiment. The direct method remains much larger over the same range of maximum step numbers.

\cref{fig:cd-required-depth} gives the required-depth comparison for the PDE observable. The direct method requires rapidly increasing $R_{\max}$ as the target accuracy is tightened. In the plotted range, it reaches $R_{\max}=32768$ for $\epsilon=10^{-6}$. The extrapolated methods need much smaller depth. The Chebyshev curve has a visible step because the selected degree changes discretely. In this data set, the $m=2$ Chebyshev interpolant has
\begin{equation}
    R_{\max}=29,
    \qquad
    E\approx 4.20\times 10^{-5}.
\end{equation}
It satisfies targets down to $\epsilon=5\times 10^{-5}$, but not the target $\epsilon=2\times 10^{-5}$. Increasing to $m=3$ gives
\begin{equation}
    R_{\max}=63,
    \qquad
    E\approx 2.07\times 10^{-7}.
\end{equation}
The jump $29\to 63$ is therefore a degree-selection effect. It is consistent with the Chebyshev relation $R_{\max}=\Theta(m^2)$, since the factor is close to $(3/2)^2$ up to rounding from the constraint $s_j=T/R_j$.

\section{Conclusion and Remarks}\label{sec:conclusion}

In this work, we use classical step-size postprocessing techniques to reduce the circuit depth of the one-ancilla quantum differential equation solver. We show that both Richardson extrapolation and Chebyshev interpolation can substantially reduce the maximum single-run circuit depth. Compared with directly using the finite-step solver, the maximum depth per run improves from $\Or(1/\epsilon)$ to $\Or(\polylog(1/\epsilon))$ in the target accuracy $\epsilon$, up to problem-dependent prefactors. While Richardson extrapolation exhibits a faster initial convergence rate, Chebyshev interpolation offers better numerical stability. Consequently, the preferred choice in practice depends on the target accuracy regime and the specific error and stability considerations of the application. Quantitatively, Richardson's extrapolation weights grow exponentially in the number of nodes, while Chebyshev's grow only polylogarithmically, after accounting for the discreteness of realizable step sizes. These scalings are consistent with the numerical experiments.

From a technical standpoint, a core ingredient distinguishing this work from existing literature on time-step extrapolation lies in handling the adjoint operators inherent to differential equations. 
We successfully addressed this challenge by introducing the holomorphic lift framework, bridging the gap between extrapolation techniques and general quantum ODE simulations.

This work demonstrates how classical postprocessing can be systematically combined with quantum circuits to mitigate quantum resource overheads for linear differential equations. Such a perspective may be naturally extended to other quantum differential equation solvers to achieve improved quantum estimates. The one step block underlying the present algorithm achieves only first order accuracy in the step size, and whether a post selected construction with higher order accuracy exists, compatible with the present extrapolation framework, remains open. Ultimately, validating the feasibility and robustness of this protocol on actual quantum hardware platforms remains an important and interesting future direction.

\section{Acknowledgement}
This work is supported by the U.S. Department of Energy, Office of Science, Accelerated Research in Quantum Computing Centers, Quantum Utility through Advanced Computational Quantum Algorithms, grant no. DE-SC0025572. We also acknowledge the support from the National Science Foundation CAREER Award DMS-2438074.

\appendix
\numberwithin{equation}{section}
\numberwithin{lemma}{section}
\section{Chebyshev Node Identities}
This appendix establishes the two trigonometric identities used to bound the Chebyshev node products and reciprocal sums in \cref{lem:chebyshevNodes}. Let $T_n$ denote the degree $n$ Chebyshev polynomial of the first kind, satisfying $T_n(\cos \alpha)=\cos(n\alpha)$, with leading coefficient $2^{n-1}$ for $n\ge 1$, and satisfying Chebyshev's differential equation
\begin{equation}
    (1-x^2)T_n''(x)-xT_n'(x)+n^2T_n(x)=0.
\end{equation}
Fix $m\ge 1$ and write $\theta_\ell:=\frac{(2\ell-1)\pi}{4m}$ for $\ell=1,\ldots,m$.

\begin{lemma}[Chebyshev Node Identities]\label{lem:chebIdentities}
    For every integer $m\ge 1$,
    \begin{equation}
        \prod_{\ell=1}^{m}\cos\theta_\ell=2^{-m+1/2}, \quad \sum_{\ell=1}^{m}\sec^2\theta_\ell=2m^2.
    \end{equation}
\end{lemma}
\begin{proof}
    The roots of $T_{2m}$ are $\cos\alpha_k$ for $\alpha_k=\frac{(2k-1)\pi}{4m}$, since $T_{2m}(\cos\alpha)=\cos(2m\alpha)=0$ exactly at these angles. For $k=1,\ldots, m$, $\alpha_k=\theta_k\in(0,\pi/2)$. For $k=m+1,\ldots, 2m$, writing $k=2m+1-\ell$ with $\ell=1,\ldots,m$,
    \begin{equation}
        \alpha_{2m+1-\ell}=\frac{(2(2m+1-\ell)-1)\pi}{4m}=\pi-\frac{(2\ell-1)\pi}{4m}=\pi-\theta_\ell,
    \end{equation}
    so $\cos \alpha_{2m+1-\ell}=-\cos\theta_\ell$. The $2m$ roots of $T_{2m}$ therefore split into $m$ pairs $\{\cos\theta_\ell,-\cos\theta_\ell\}_{\ell=1}^{m}$, and $T_{2m}$ is even, with
    \begin{equation}
        T_{2m}(x)=2^{2m-1}\prod_{\ell=1}^{m}(x-\cos\theta_\ell)(x+\cos\theta_\ell)=2^{2m-1}\prod_{\ell=1}^{m}(x^2-\cos^2\theta_\ell).
    \end{equation}

    \textbf{Product of Nodes.} Setting $x=0$,
    \begin{equation}
        T_{2m}(0)=2^{2m-1}\prod_{\ell=1}^{m}(-\cos^2\theta_\ell)=2^{2m-1}(-1)^m\prod_{\ell=1}^{m}\cos^2\theta_\ell.
    \end{equation}
    On the other hand, $T_{2m}(0)=\cos(2m\cdot \frac{\pi}{2})=\cos(m\pi)=(-1)^m$. Equating the two expressions and canceling $(-1)^m$ on both sides gives
    \begin{equation}
        \prod_{\ell=1}^{m}\cos^2\theta_\ell=2^{-2m+1}.
    \end{equation}
    Since each $\theta_\ell\in(0,\pi/2)$, every $\cos\theta_\ell>0$, so taking the positive square root gives $\prod_\ell\cos\theta_\ell=2^{-m+1/2}$.

    \textbf{Sum of Reciprocal Squared Nodes.} Taking the logarithmic derivative of $T_{2m}$ gives
    \begin{equation}
        \frac{T_{2m}'}{T_{2m}}=\sum_{\ell=1}^{m}\frac{2x}{x^2-\cos^2\theta_\ell}.
    \end{equation}
    Differentiating once more,
    \begin{equation}
        \left(\frac{T_{2m}'}{T_{2m}}\right)'=\sum_{\ell=1}^{m}\frac{-2(x^2+\cos^2\theta_\ell)}{(x^2-\cos^2\theta_\ell)^2}.
    \end{equation}
    At $x=0$, this becomes $(\frac{T_{2m}'}{T_{2m}})'(0)=-2\sum_{\ell=1}^{m}\frac{1}{\cos^2\theta_\ell}=-2\sum_{\ell=1}^{m}\sec^2\theta_\ell$. Since $T_{2m}$ is even, $T_{2m}'(0)=0$, and therefore, 
    \begin{equation}
        \left(\frac{T_{2m}'}{T_{2m}}\right)'(0)=\frac{T_{2m}(0)T_{2m}''(0)-T_{2m}'(0)T_{2m}'(0)}{T_{2m}^2(0)}=\frac{T_{2m}''(0)}{T_{2m}(0)}.
    \end{equation}
    Evaluating Chebyshev's differential equation for $T_{2m}$ at $x=0$,
    \begin{equation}
        T_{2m}''(0)+4m^2T_{2m}(0)=0.
    \end{equation}
    This implies $-4m^2=\frac{T_{2m}''(0)}{T_{2m}(0)}=(\frac{T_{2m}'}{T_{2m}})'(0)=-2\sum_{\ell=1}^{m}\sec^2\theta_\ell$, so $\sum_{\ell=1}^{m}\sec^2\theta_\ell=2m^2$.
\end{proof}

\bibliographystyle{unsrt}

\bibliography{ref_all}

\end{document}